\documentclass[10pt]{article}
\usepackage{array,url}
\usepackage[margin=1.2in]{geometry}
\usepackage[geometry]{ifsym}
\usepackage[toc,page]{appendix}
\usepackage{amssymb, amsfonts, amsmath, amsthm, amscd, mathrsfs, mathtools, color}
\usepackage{todonotes, comment, verbatim, fullpage}
\usepackage{graphicx, float, epsfig, enumerate, fixmath, bbm, latexsym, cite, scrextend}
\usepackage{pgfplots}
\pgfplotsset{compat=1.12}
\usepackage{subcaption}
\usepackage{dsfont}

\usepackage{xcolor}

\usepackage{mleftright}  % fixing error with respect to the spacing of \left and \right
\mleftright

\newcommand{\bbE}{\mathbb{E}}\newcommand{\rme}{\mathrm{e}}

\newcommand{\bbI}{\mathbb{I}}

\newcommand{\bbN}{\mathbb{N}}\newcommand{\rmN}{\mathrm{N}}

\newcommand{\bbP}{\mathbb{P}}

\newcommand{\bbR}{\mathbb{R}}

\newcommand{\sfA}{\mathsf{A}}

\newcommand{\sfD}{\mathsf{D}}

\newcommand{\sfP}{\mathsf{P}}

\newcommand{\cI}{\mathcal{I}}

\newcommand{\scrN}{\mathscr{N}}

\newcommand{\scrS}{\mathscr{S}}

\newtheoremstyle{mystyle}% name
{}% space above
{}% space below
{\itshape}% body font
{}% indent amount
{\bfseries}% theorem head font
{}% punctuation after theorem head
{.5em}% space after theorem head
{}% theorem head spec

\newtheoremstyle{remark}% name
{}% space above
{}% space below
{}% body font
{}% indent amount
{\itshape}% theorem head font
{}% punctuation after theorem head
{.5em}% space after theorem head
{}% theorem head spec

\makeatletter
\def\thmhead@plain#1#2#3{%
  \thmname{#1}\thmnumber{\@ifnotempty{#1}{ }\@upn{#2.}}%
  \thmnote{ \textsf{\the\thm@notefont\textit{#3}.}}}
\let\thmhead\thmhead@plain
\makeatother

\theoremstyle{mystyle}
\newtheorem{theorem}{Theorem}%[section]
\theoremstyle{mystyle}
\newtheorem{lemma}{Lemma}%[section]
\theoremstyle{mystyle}
\newtheorem{prop}{Proposition}%[section]
\theoremstyle{mystyle}
\theoremstyle{mystyle}
\newtheorem{definition}{Definition}%[section]
\theoremstyle{remark}
\newtheorem{rem}{Remark}%[section]
\theoremstyle{mystyle}
\theoremstyle{mystyle}
\theoremstyle{mystyle}
\theoremstyle{discussion}
\theoremstyle{mystyle}
\theoremstyle{mystyle}
\usepackage{tocloft} % THIS PACKAGE DOES NOT WORK WITH \usepackage{subfigure} 
\usepackage{titlesec}
\titleformat{\section}
  {\normalfont\large\bfseries}{\thesection.}{0.5em}{}
 
\titleformat{\subsection}
  {\normalfont\bfseries}{\thesubsection.}{0.5em}{}
 
\renewcommand{\thesection}{\arabic{section}}
\renewcommand{\thesubsection}{\thesection.\arabic{subsection}}
\newcommand\independent{\protect\mathpalette{\protect\independent}{\perp}}
\def\independent#1#2{\mathrel{\rlap{$#1#2$}\mkern2mu{#1#2}}}

\def\squarebox#1{\hbox to #1{\hfill\vbox to #1{\vfill}}}

\makeatletter
\newcommand{\vast}{\bBigg@{4}}
\newcommand{\Vast}{\bBigg@{5}}
\newcommand{\Gigantic}{\bBigg@{8}}

\makeatother
\allowdisplaybreaks

\newcommand{\supp}{{\mathsf{supp}}}

\title{Properties of the Support of the Capacity-Achieving Distribution of the Amplitude-Constrained Poisson Noise Channel}
\author{Alex Dytso\thanks{A. Dytso is with the Department of Electrical and Computer Engineering, New Jersey Institute of Technology, Newark,  NJ 07102, USA   (e-mail: alex.dytso@njit.edu).
},  
Luca Barletta\thanks{L. Barletta  is with the Dipartimento di Elettronica, Informazione e Bioingegneria, Politecnico di Milano, Milano, 20133, Italy. (e-mail: luca.barletta@polimi.it).},
 and Shlomo Shamai (Shitz)%
\thanks{The work of S. Shamai has been supported by the European Union's Horizon
2020 Research And Innovation Programme, grant agreement no. 694630, and by
the United States-Israel Binational Science Foundation under Grant
BSF-2018710. }
}
\begin{document}

\maketitle

\begin{abstract}
This work considers a Poisson noise channel with an amplitude constraint. It is well-known that the capacity-achieving input distribution for this channel is discrete with finitely many points.  We sharpen this result by introducing upper and lower bounds on the number of mass points. Concretely, an upper bound of order $\sfA \log^2(\sfA)$ and a lower bound of order $\sqrt{\sfA}$ are established where $\sfA$ is the constraint on the input amplitude. In addition, along the way, we show several other properties of the capacity and capacity-achieving distribution. For example, it is shown that the capacity is equal to  $ - \log P_{Y^\star}(0)$ where $P_{Y^\star}$ is the optimal output distribution. Moreover, an upper bound on the values of the probability masses of the capacity-achieving distribution and a lower bound on the probability of the largest mass point are established.    Furthermore, on the per-symbol basis, a nonvanishing  lower bound on the probability of error for detecting the capacity-achieving distribution is established under the maximum a posteriori rule. 

\end{abstract}

\section{Introduction} 
\label{sec:Introduction}
We consider a discrete-time memoryless Poisson channel. The output $Y$ of this channel   takes value on the set of nonnegative integers $\mathbb{N}_0$ and the input $X$ takes value on the set of nonnegative real numbers.   The conditional probability mass function (pmf) of the output random variable $Y$ given the input $X$ that  specifies the channel is given by 
\begin{equation}
P_{Y|X}(y|x) = \frac{1}{y!} \, x^y \, \rme^{- x}, \quad x\ge 0, \, y\in \mathbb{N}_0. \label{eq:PoissonTransformation}
\end{equation}
In \eqref{eq:PoissonTransformation}, we use the standard convention that $0^0=1$ and $0!=1$. 

The capacity of this channel where the input $X$ is subject to the amplitude constraint $0\le X \le \sfA$ is given by 
\begin{equation}
C(\sfA)= \max_{ X :\:  0 \le X \le \sfA } I(X;Y), \quad \sfA>0. \label{eq:Def_Capacity}
\end{equation} 
Finding the  capacity of this channel remains to be an elusive task.  
The goal of this work is to make progress on this problem by studying the properties of 
the capacity-achieving distribution denoted by   $P_{X^\star}$. 

\subsection{Contributions and Paper Outline}
The outline and the contribution of the paper are as follows.  The remaining part of Section~\ref{sec:Introduction} will survey the relevant literature, present our notation, and will go over the key tools needed in our analysis, such as the oscillation theorem and the strong data-processing inequality. Section~\ref{sec:Main_Results} presents our main results, which include the following: a new  compact representation of the capacity; an upper bound on the values of probability masses of the optimal input distribution $P_{X^\star}$; a lower bound on the probability  of the largest mass points (i.e.,  $P_{X^\star}(\sfA)$); a lower bound on the size of the support of $P_{X^\star}$;  an upper bound on the size of the support of $P_{X^\star}$;   a nonvanishing  lower bound,  on the per-symbol basis, on the probability of error for detecting the capacity-achieving distribution. Section~\ref{sec:proofs} is dedicated to the proofs.   Section~\ref{sec:Conclusion} concludes the paper with some final remarks and possible future directions.  
%\item  
%\begin{itemize}
%\item Section~\ref{sec:Main_Results} presents our main results, which include the following: a new  compact representation of the capacity; an upper bound on the values of probability masses of the optimal input distribution $P_{X^\star}$; a lower bound on the probability  of the largest mass points (i.e.,  $P_{X^\star}(\sfA)$); a lower bound on the size of the support of $P_{X^\star}$; and an upper bound on the size of the support of $P_{X^\star}$. 
%\item Section~\ref{sec:proofs} is dedicated to the proofs. 
%\item 
% \end{itemize}

\subsection{Prior Work}

 The now classical approach developed by Witsenhausen in \cite{witsenhausen1980some} says that if the output alphabet has a cardinality $n$, then the support of the optimal input distribution cannot be more than $n$, irrespective of the size of the input alphabet.  
However, since the output alphabet of the Poisson noise channel has a countably infinite alphabet, the Witsenhausen approach does not apply. Instead, the approach that has been applied to the Poisson noise channel largely follows the analyticity idea introduced by Smith in \cite{smith1971information} in the context of amplitude-constrained Gaussian noise channel. The interested reader is referred to \cite{dytso2018discrete} for a summary of these techniques.  In this work, we also follow the latter technique. However, we considerably generalize and improve this approach.  In what follows, we summarize the known results on the Poisson noise channel and highlight the elements of the new technique.

The discrete-time Poisson noise channel is suited to model low-intensity, direct-detection optical communication channels \cite{gordon1962quantum}; the interested reader is also referred to  a survey on  free-space optical communications in \cite{khalighi2014survey}. The Poisson channel can be seen as a limiting case of the Binomial channel \cite{komninakis2001capacity}, which can be used to model the number of particles released by a sender unit in molecular communications \cite{Farsad2020}. A key difference in the mathematical formulation between the Poisson and the Binomial channel models is the infinite/finite nature of the output alphabet. In this work, we are concerned with the discrete-time channel; however, there also exists  a large literature on continuous-time channels, and the interested reader is referred to a survey in \cite{verdu1999poisson}.

The first major study of the capacity-achieving distribution for the Poisson channel was undertaken  in \cite{mceliece1979practical},  where the authors considered the capacity    with and without an additional  average-power constraint on the input,\footnote{In the Poisson noise channel the average-power is measured in terms the first moment of the input $X$.  The term power is appropriate if we consider laser communications where $\bbE[X]=\sfP$ represent average photon count. In this case, $ \sfP=\frac{ \eta P_t}{h f_0}$ where $P_t$ is the transmit power, $f_0$ is the optical carrier frequency, $h$ is Plank's constant and $\eta$ is the proportionality constant \cite{gordon1962quantum,pierce1981capacity}.} that is
\begin{equation}
 \bbE[X]\le \sfP.
\end{equation}
The authors of   \cite{mceliece1979practical} derived the Karush-Kuhn-Tucker (KKT) conditions that are necessary to study the structure of the capacity-achieving distribution.   These KKT conditions were then used to show that  the support of an optimal input distribution for any $\mathsf{A}$  can contain at most one mass point in the  open interval $(0,1)$. 
Moreover,  for any $\sfA\le 1$, it was shown that the optimal input distribution consists of two mass points at $0$ and $\sfA$  and is given by 
\begin{equation}
P_{X^\star}(\sfA)=\frac{1}{{\mathrm{e}}^{\frac{\sfA}{{\mathrm{e}}^\sfA-1}}-{\mathrm{e}}^{-\sfA}+1},\quad
P_{X^\star}(0)=1-P_{X^\star}(\sfA), \label{eq:OptimazingDistributionSmallAmplitude}
\end{equation} 
and the capacity is given by 
\begin{equation}
C(\sfA)= - \log \left( P_{X^\star}(0)+\rme^{-\sfA}P_{X^\star}(\sfA)   \right). \label{eq:ShamaiCapacity}
\end{equation}

The KKT conditions proposed in   \cite{mceliece1979practical} were rigorously derived   in \cite[Corollary 1]{shamai1990capacity} and extended to a more general case that includes the possibility of a nonzero dark current\footnote{A more general Poisson noise model also incorporates a nonnegative parameter $\lambda$ known as the dark current parameter, and  the channel is given  by $P_{Y|X}(y|x) = \frac{1}{y!} (x+\lambda)^y \rme^{- (x+\lambda)},\, x\ge 0,y\in \mathbb{N}_0 $. The dark current parameter $\lambda$ represents the intensity of an additional source of noise or interference, which produces  on average extra $\lambda$ photons   at a particle counter \cite{gordon1962quantum,pierce1981capacity, verdu1990asymptotic,shamai1990capacity}.} parameter. Moreover, using the analyticity idea of \cite{smith1971information}, in  \cite{shamai1990capacity} for $\sfA<\infty$ and any  $\sfP>0 $ it was shown that  the optimizing input distribution is unique and  discrete with finitely many mass points. Moreover,  for the case of $\sfP\ge \sfA$  (i.e., the average-power constraint is not active) and dark current is zero,  it  was shown that the distribution in \eqref{eq:OptimazingDistributionSmallAmplitude} continues to be capacity-achieving if and only if   $\sfA \le  \bar{\sfA}$ where $\bar{\sfA}\approx 3.3679$.   

Further studies of the conditions under which the capacity-achieving distribution is binary have been undertaken by the authors of \cite{cao2014capacity} and \cite{cao2014capacityPart2}.   For example, in \cite{cao2014capacity},  it was shown that with both the amplitude and the average-power constraint,   the optimal input distribution always contains a mass point at $0$. Moreover, in the case of only an amplitude constraint,  the optimal input distribution  contains mass points at both $0$ and~$\sfA$.  In \cite{cao2014capacityPart2}, it was shown that if  $\sfP < \frac{\sfA}{2}$ and the dark current  is large enough, the following binary distribution is optimal:
   \begin{equation}
   P_{X^\star}( 0) = 1 -\frac{\sfP}{\sfA} ,   \quad P_{X^\star}( \sfA) = \frac{\sfP}{\sfA}. 
   \end{equation}
   The capacity-achieving distribution with only an average-power constraint was  considered in \cite{cheraghchi2018improved}  and was shown to be discrete with infinitely many mass points.  
   
   The low-average-power and the low-amplitude asymptotics of the capacity have been studied in \cite{martinez2007spectral,martinez2008achievability,lapidoth2011discrete,wang2014refined,wang2014impact}. A number of papers have also focused on upper and lower bounds on the capacity. The first upper and lower bounds  on the capacity have been derived in  \cite{mceliece1979practical} for two situations: the case of the average-power constraint only, and the case of both the average-power and the amplitude constraint with $\sfA\le 1$. 
The authors of  \cite{verdu1990asymptotic}  derived upper and lower bounds, in the case of   the average-power constraint only, by focusing on the regime where both $\sfP$ and the dark current tend to infinity with a fixed ratio. 
Firm upper and lower bounds on the capacity in the case of only the average-power constraint and no dark current have been derived  in \cite{martinez2007spectral} and \cite{martinez2008achievability}.  Bounds in \cite{martinez2007spectral} and \cite{martinez2008achievability} have been further improved in \cite{cheraghchi2018improved} and \cite{cheraghchi2020non}. The most general bounds on the capacity  that consider both the amplitude and the average-power constraints on the input and hold for an arbitrary value of the dark current have been derived  in  \cite{lapidoth2009capacity}.   The bounds  in \cite{lapidoth2009capacity} have been shown to be  tight in the regime where both the average-power and the amplitude constraint approach infinity with a fixed ratio $\frac{\sfP}{\sfA}$.  Finally, the authors of \cite{yu2014lower} sharpened the results of \cite{lapidoth2009capacity} for small values of $\sfP$ and $\sfA$.

\subsection{Notation} 
Throughout the paper, the deterministic scalar quantities are denoted by lower-case letters and random variables are denoted by uppercase letters.   The set of all real numbers, nonnegative integers and positive integers are denoted by  $\mathbb{R}, \mathbb{N}_0$, and $\mathbb{N}$, respectively.     The  function $\mathds{1}_{\mathcal{A}}(x)$ denotes the indicator function over the set $\mathcal{A}$ where $\mathds{1}_{\mathcal{A}}=1$ if $x \in \mathcal{A}$ and  $\mathds{1}_{\mathcal{A}}=0$ otherwise. 

We denote the distribution  of a random variable $X$ by $P_{X}$. The support set of $P_X$ is denoted and defined as
\begin{equation}
\supp(P_{X})= \left\{x:  \text{ for every open set $ \mathcal{D} \ni x $ we have that $P_{X}( \mathcal{D})>0$}  \right\}. 
\end{equation} 
The relative entropy between distributions $P$ and $Q$ will be denoted by $\sfD(P\|Q)$. 
A Poisson pmf with expected value $x$ is denoted by ${\cal P}(x)$.

The number of zeros of a function $f \colon \mathbb{R} \to \mathbb{R} $  on the interval $\cI$ is denoted by  $\rmN(\cI, f)$. Let $g_1$ and $g_2$ be nonnegative functions, then  
\begin{itemize}
\item $g_1(x)=O(g_2(x))$ means that there exists a constant $c>0$ and $x_0$ such that $\frac{g_1(x)}{g_2(x)}\le c$ for all $x>x_0$;
\item $g_1(x)=\Omega(g_2(x))$ means that $g_2(x)=O(g_1(x))$; and
\item $g_1(x)=o(g_2(x))$ means $ \lim_{x \to \infty} \frac{g_1(x)}{g_2(x)}=0$. 
\end{itemize}

Finally,  the Lambert W-function is denoted by  $W_k$ where $k$ indicates the branch \cite{NIST:DLMF}. Since we are dealing with real numbers, we only use the principle branch $W_0$ and the lower branch $W_{-1}$ .   Recall that, for the real numbers $x$ and $y$,  the Lambert W-function provides a solution to the equation $y \rme^y=x$, which can be solved for $y$ only if $x \ge \rme^{-1}$. Moreover, in the regime $x \ge 0$,  the solution is unique and is given by $y=W_0(x)$, and in the regime $\rme^{-1} \le x<0$,  there are two solutions give by $y=W_0(x)$ and $y=W_{-1}(x)$.

 \subsection{Overview of the Key Tools} 
 In this section, we overview the main tools needed in our analysis.

 \paragraph{Strong Data-Processing Inequality}   In order to find an upper bound on the values of probability masses and a lower bound on the number of mass points, we will rely on the strong data-processing inequality for the relative entropy. The study of the strong data-processing inequalities has recently received some attention, and the interested reader is referred to \cite{RaginskyStrongData,polyanskiy2017strong,du2017strong} and references therein.

 Fix some channel $Q_{Y|X}:  \mathcal{X} \to \mathcal{Y}$ where $\mathcal{X}$ and $\mathcal{Y}$ are the input and output alphabets, respectively.   Let $Q_Y$ denote the distribution on $\mathcal{Y}$ induced by the push-forward of the distribution $Q_X$ on $\mathcal{X}$ through $Q_{Y|X}$. We denote this operation by $Q_X\to Q_{Y|X}\to Q_Y$.    
 
  The classical \emph{data-processing inequality} for the relative entropy states that  for  $Q_X\to Q_{Y|X}\to Q_Y$  and $P_X\to Q_{Y|X}\to P_Y$, we have 
 \begin{equation}
 \sfD(P_Y \| Q_Y) \le \sfD(P_X\| Q_X).  \label{eq:DT_KL_Classical}
 \end{equation} 
 The \emph{strong data-processing inequality} is the sharpening of \eqref{eq:DT_KL_Classical}  where one seeks to find a coefficient  $0 < \eta \le 1$ such that 
 \begin{equation}
  \sfD(P_Y \| Q_Y) \le  \eta \sfD(P_X\| Q_X).  \label{eq:SDP_def}
 \end{equation}
It is not difficult to see that
	\begin{equation}
	\eta \le 	 \eta_{\textsf{\textup{KL}}}(\mathcal{X};Q_{Y|X} ) \coloneqq\sup_{\substack{Q_X, P_X:  \\ \supp(Q_{X}) \subseteq \mathcal{X},   \supp(P_{X}) \subseteq \mathcal{X} \\ 0<D(Q_X\|P_{X})<\infty\\ Q_X\to Q_{Y|X}\to Q_Y \\ P_X\to Q_{Y|X}\to P_Y}}\frac{D(Q_Y\|P_{Y})}{D(Q_X\|P_{X})}.
		 \label{eq:Defintion_Of_Contraction_Coefficient}
	\end{equation}%
	The quantity $0< \eta_{\textsf{\textup{KL}}}(\mathcal{X};Q_{Y|X} ) \le 1 $ is known as the contraction coefficients. In Section~\ref{sec:Proof_bounds_On_Probabilities}, for the case of Poisson channel, we will provide a nontrivial upper bound on $ \eta_{\textsf{\textup{KL}}}(\mathcal{X},Q_{Y|X} )$.

%The strong data-processing inequality will be the key tool used provided bounds on the values of the probability masses of the capacity-achieving distribution. 

 \paragraph{Oscillation Theorem}    To find an upper bound on the number of points in the support of $P_{X^\star}$, we will follow the proof technique developed in \cite{dytso2019capacity} for the Gaussian noise channel.   The key step that we borrow from \cite{dytso2019capacity} is the use of the variation-diminishing property, which is captured by the oscillation theorem of Karlin \cite{karlin1957polya}.     To state the oscillation theorem we need the following definition. 
 
\begin{definition}[Sign Changes of a Function]  The number of sign changes of a function $\xi: \mathcal{X} \to \mathbb{R}$ is given by 
\begin{equation}
  \scrS(\xi) = \sup_{m\in \bbN } \left\{\sup_{y_1< \cdots< y_m  \subseteq \mathcal{X} } \scrN \{ \xi (y_i) \}_{i=1}^m\right\} \text{,}
\end{equation}
where  $\scrN\{ \xi (y_i) \}_{i=1}^m$ is the number of sign changes  of the sequence $\{ \xi (y_i) \}_{i=1}^m $.
\end{definition}

The following theorem, shown in \cite[Thm.~3]{karlin1957polya} (see also  \cite[Thm.~3.1,~p.~21]{KarlinBook1968}), will be a key step in the proof of the upper bound on the number of mass points.

\begin{theorem}[Oscillation Theorem]\label{thm:OscillationThoerem} Given domains $\bbI_1 $ and $\bbI_2$, let $p\colon \bbI_1\times \bbI_2  \to \bbR$ be a strictly totally positive  kernel.\footnote{A function $f:\bbI_1 \times \bbI_2 \to \bbR$ is said to be a strictly totally positive  kernel of order $n$ if $\det\left([f(x_i,y_j)]_{i,j = 1}^{m}\right) >0 $ for all $1\le m \le n $, and for all $x_1< \cdots < x_m \in \bbI_1  $, and $y_1< \cdots < y_m \in \bbI_2$.  If $f$ is a strictly totally positive  kernel  of order $n$ for all $n\in \bbN$, then $f$ is a strictly totally positive  kernel.} For an arbitrary $y$, suppose $p(\cdot, y)\colon \bbI_1 \to \bbR $ is an $n$-times differentiable function. Assume that $\mu$ is a  regular sigma-finite measure on $\bbI_2 $, and let $\xi \colon \bbI_2 \to \bbR $ be a function with $\scrS(\xi) = n$. For $x\in \bbI_1$, define
\begin{equation}
\Xi(x)=  \int  \xi (y)\: p(x ,y)  \, {\rm d} \mu(y) \text{.} \label{eq:Integral_Transform}
\end{equation}
If $\Xi \colon \bbI_1 \to \bbR$ is an $n$-times differentiable function, then either $\rmN(\bbI_1, \Xi) \le n$, or $\Xi\equiv 0$.  
\end{theorem} 

The above theorem says that the number of zeros of a function $\Xi(x)$, which is the output of integral transformation, is less than the number of sign changes of the function $  \xi (y) $, which is the input to the integral transformation.   
In the rest of this paper, we will take the kernel  $ p(x ,y) $ to be the Poisson transition probability $P_{Y|X}$ in \eqref{eq:PoissonTransformation}, $\mu$ to be the counting measure, and  the domains will be set to $\bbI_1 =\{ x: x \ge 0\}$ and  $\bbI_2 = \mathbb{N}_0$. The fact that the Poisson transition probability is a  strictly totally positive  kernel was shown in \cite{karlin1957polya}.  The interested reader is also referred to \cite[p.~19]{KarlinBook1968},  where it is shown that members of the exponential family (e.g., Poisson, Gaussian, chi-squared)  are positive definite kernels.   Figure~\ref{fig:example_Oscillation_THm} shows an example of Theorem~\ref{thm:OscillationThoerem} for the Poisson kernel. 

\begin{figure}
\center
\begin{subfigure}{.4\linewidth}
\definecolor{mycolor1}{rgb}{0.00000,0.44700,0.74100}%
\definecolor{mycolor2}{rgb}{0.85000,0.32500,0.09800}%
\begin{tikzpicture}

\begin{axis}[%
width=6cm,
height=5cm,
at={(1.011in,0.642in)},
scale only axis,
xmin=0,
xmax=40,
xlabel={$y$},
ymin=-2,
ymax=2,
axis background/.style={fill=white},
xmajorgrids,
ymajorgrids,
legend style={legend cell align=left, align=left, draw=white!15!black}
]

\addplot[ycomb, color=red, mark=o, mark options={solid, red}] table[row sep=crcr] {%
0	0\\
1	1.105341068072\\
2	1.33074857582821\\
3	0.663122658240795\\
4	-0.0430415378907821\\
5	0.0689908389009757\\
6	0.992708874098054\\
7	1.85873427788249\\
8	1.80104164646985\\
9	0.822983865893657\\
10	-0.202902745543642\\
11	-0.40130223174067\\
12	0.239315664287558\\
13	0.866025403784438\\
14	0.626709739496882\\
15	-0.464723172043766\\
16	-1.52914806202523\\
17	-1.6890092696781\\
18	-0.935016242685415\\
19	-0.126683470313616\\
20	-0.126683470313614\\
21	-0.935016242685414\\
22	-1.68900926967809\\
23	-1.52914806202524\\
24	-0.46472317204377\\
25	0.626709739496879\\
26	0.86602540378444\\
27	0.239315664287562\\
28	-0.40130223174067\\
29	-0.202902745543646\\
30	0.822983865893651\\
31	1.80104164646985\\
32	1.85873427788249\\
33	0.992708874098059\\
34	0.0689908389009792\\
35	-0.0430415378907851\\
36	0.663122658240794\\
37	1.33074857582821\\
38	1.105341068072\\
39	5.51214855743627e-15\\
40	-1.10534106807199\\
};
\addplot[forget plot, color=white!15!black] table[row sep=crcr] {%
0	0\\
40	0\\
};
\addlegendentry{$\xi(y)$}

\end{axis}

\end{tikzpicture}%
\caption{Input function $\xi(y)$.}
\end{subfigure}
~
\hspace{2cm}
\begin{subfigure}{.4\linewidth}
\input{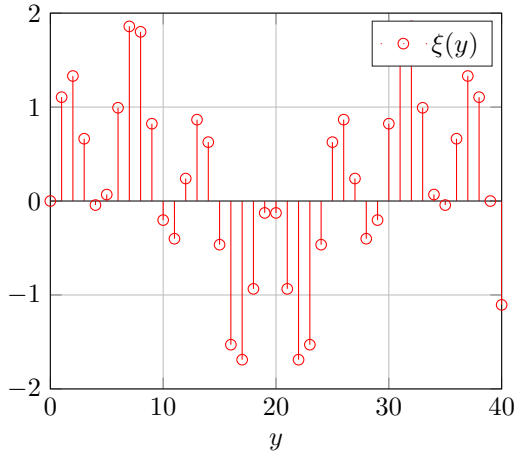}
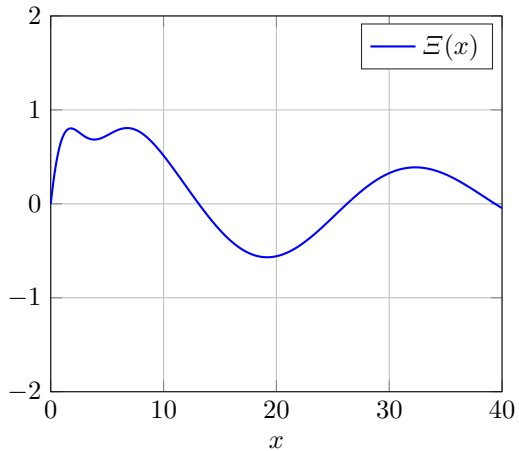
\caption{Output function $\Xi(y)$.}
\end{subfigure}
\caption{An example of the oscillation theorem with the Poisson kernel where we set  $\xi(y)=\sin\left( \frac{\pi}{3} y \right) +\sin \left( \frac{\pi}{13} y\right), \, y \in  \bbN_0 $.  }
\label{fig:example_Oscillation_THm}
\end{figure} 

The oscillation theorem allows to upper bound the number of points in the support of $P_{X^\star}$ with the number of sign changes of a function that is related to the output distribution $P_{Y^\star}$. 
 In the case of the Gaussian noise channel, in order to count the number of sign changes, one needs to resort to complex analytic techniques.  In contrast,  in the Poisson case, due to the discrete nature of the channel, we no longer need to rely on the complex analytic techniques, which simplifies this part of the analysis.
 
However, the analysis for the Poisson channel is not necessarily simpler than that in the Gaussian case. One crucial step in both proofs relies on finding a lower bound of the output pdf in the Gaussian case and the output pmf in the Poisson case.  In the Gaussian case, this can be done by using Jensen's inequality, and the resulting bound is universal and is independent of $P_{X^\star}$.  In the Poisson noise case, however, a distribution-independent bound cannot be obtained, and the lower bound on the tail depends on  $P_{X^\star}$. Specifically, the lower bound depends on  the value of $P_{X^\star}(\sfA)$. Therefore, to complete the proof, we need to also find a lower bound on $P_{X^\star}(\sfA)$. 

\paragraph{Connection Between Estimation and Information Measures} 
Some of the key intermediate steps in our proofs will rely on identities that connect  information measures and estimation measure. In particular, we will rely on the expression for the conditional expectation derived in \cite{dytso2020estimation}. 
For further connections between estimation and information theoretic measure the interested reader is referred to \cite{guo2008mutual} and \cite{atar2012mutual} and references therein.

\section{Main Results} 
\label{sec:Main_Results}

The main results of this paper are summarized in the following theorem.

\begin{theorem}\label{thm:Main_Result}    The capacity and the capacity-achieving  distribution $P_{X^\star}$ of an amplitude-constrained Poisson noise channel satisfy the following properties: 
\begin{itemize}
\item \emph{A New Capacity Expression:} For every $\sfA  \ge 0 $ , the capacity is given by 
\begin{equation}
C(\sfA) =  \log \frac{1}{P_{Y^\star}(0)}, \label{eq:New_Cap_Expression}
\end{equation}
where $P_{Y^\star}$ is the capacity-achieving output distribution induced by $P_{X^\star}$. 
\item \emph{An Upper Bound on the Probabilities:} For every $\sfA  \ge 0$,
\begin{align}
\text{Universal Bound: }P_{X^\star}(x) &\le \rme^{- \frac{C(\sfA)}{ 1-\rme^{-\sfA}}  }, &  & x\in \supp(P_{X^\star}), \label{Eq:Bound_On_Probability_Values_poisson} \\
\text{Location Dependent Bound: } P_{X^\star}(x) &\le \rme^{-C(\sfA)-x \frac{\rme^{-x}}{1-\rme^{-x}}},&  &  x\in \supp(P_{X^\star}) \setminus \{0\}.  \label{Eq:Bound_On_Probability_Values_poisson_loc_dependent}
\end{align}
In addition, if   $| \supp(P_{X^\star}) |=2$, then  the bound in \eqref{Eq:Bound_On_Probability_Values_poisson_loc_dependent} becomes equality. 
\item   \emph{A Lower Bound on the Probability of the Largest Point:} For all $\sfA$ such that 
$\rme^{ \frac{C(\sfA)}{ 1-\rme^{-\sfA}}  } \ge 4$, we have that 
\begin{equation}
  P_{X^\star}(\sfA) \ge    \frac{1-3\:\rme^{- \frac{C(\sfA)}{ 1-\rme^{-\sfA}}  } }{2\sfA^{2\sfA \rme  \log(\sfA)+2} \, \rme^{-2\sfA+1}  } .  \label{eq:lower_bound_on_P_X(A)}
\end{equation}
\item \emph{On the Location of Support Points:}

Suppose that $ | \supp(P_{X^\star}) | \ge 3$ and let  $\supp(P_{X^\star})\setminus \{0,\sfA\}$.  Then,
\begin{equation}
 \rme^{-\sqrt{2(\log(\sfA)-1)}} \le \sfA \, \rme^{W_{-1} \left(- \frac{1}{\sfA} \right)} \le x^\star  \le \sfA \, \rme^{W_{0} \left(- \frac{1}{\sfA} \right)} \le \sfA-1.   \label{eq:Bound_on_other_points}
\end{equation}

\item \emph{A Lower Bound on the Size of the Support:} For every $\sfA  \ge 0$,
\begin{equation}
    | \supp(P_{X^\star}) |  \ge    \rme^{ \frac{C(\sfA)}{  1-\rme^{- \sfA } }  }.  \label{eq:LowerBound_on_Number_of_Points}
\end{equation}
\item  \emph{An Upper Bound on the Size of the Support:} For every $\sfA \ge 0$,
 \begin{equation}
  | \supp(P_{X^\star}) |  \le  \left \lceil \sfA  - \log  \left(   P_{X^\star}(\sfA) \right) - C(\sfA) \right \rceil +2. \label{eq:First_implicit_bound_on_point}
  \end{equation}
  In addition, for all $\sfA$ such that 
$\rme^{ \frac{C(\sfA)}{ 1-\rme^{-\sfA}}  } \ge 4$, we have that 
\begin{equation}
     | \supp(P_{X^\star}) |     \le 2  \rme \, \sfA  \log^2(\sfA)+2\log(\sfA)-\sfA- \log  \left(  \frac{1-3\:\rme^{- \frac{C(\sfA)}{ 1-\rme^{-\sfA}}  } }{2 } \right) - C(\sfA)  +4.  \label{eq:Bounds_on_the_Number_Of_Poitns}
\end{equation}

\end{itemize} 

\end{theorem} 
The proof of Theorem~\ref{thm:Main_Result} is given in Section~\ref{sec:proofs}. A few comments are now in order.

\subsection{Numerical Simulations} 
In order to aid our discussion, we have also numerically computed the optimal input distributions for values of $\sfA$ up to $15$. 
Fig.~\ref{fig:Plot_of_probabilitiues} depicts the output of this simulation.    

We note that there are several numerical recipes for generating an optimal input distribution \cite{blahut1972computation,chang1988calculating,huang2005characterization}.  However, most of these approaches  ultimately optimize over the space of distributions, which is an infinite-dimensional space.  As was already alluded to in \cite{smith1971information} and \cite{dytso2019capacity},  a firm upper bound on the number of mass points, such as the one in  Theorem~\ref{thm:Main_Result}, allows us to move the optimization  from the space of probability distributions  to the space $\mathbb{R}^{2n}$ where $n$ is the number of points.   Working in $\mathbb{R}^{2n}$ allows us to employ methods such as the projected gradient ascent \cite{shalev2014understanding}, which was used to generate the plots in Fig.~\ref{fig:Plot_of_probabilitiues}. A quick sketch of how to implement the projected gradient ascent is given in Appendix~\ref{app:Gradient_Ascent}.

\begin{figure}
\center
\begin{subfigure}[t]{.4\linewidth}
\input{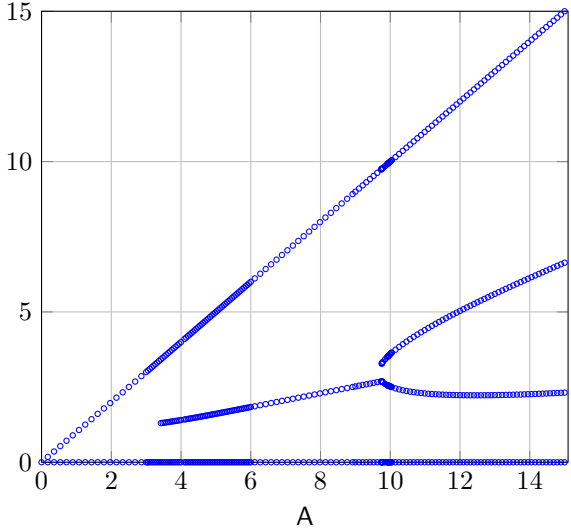}
\caption{Plot of the locations of the support points of $P_{X^\star}$ vs. $\sfA$.}
\label{fig:Locations}
\end{subfigure}
~
\hspace{1cm}
\begin{subfigure}[t]{.4\linewidth}
\input{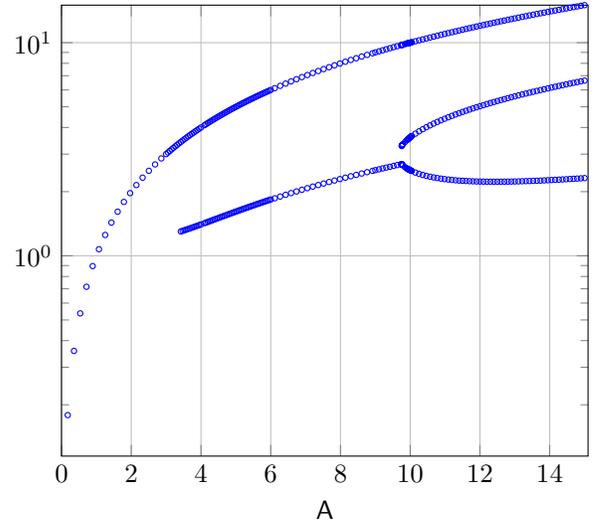}
\caption{Plot of the locations of the support points of $P_{X^\star}$ vs. $\sfA$ where the log scale is used for the vertical axis.}
\label{fig:Locations_Log_scale}
\end{subfigure}
~
\vspace{1cm}

\begin{subfigure}[t]{.4\linewidth}
\input{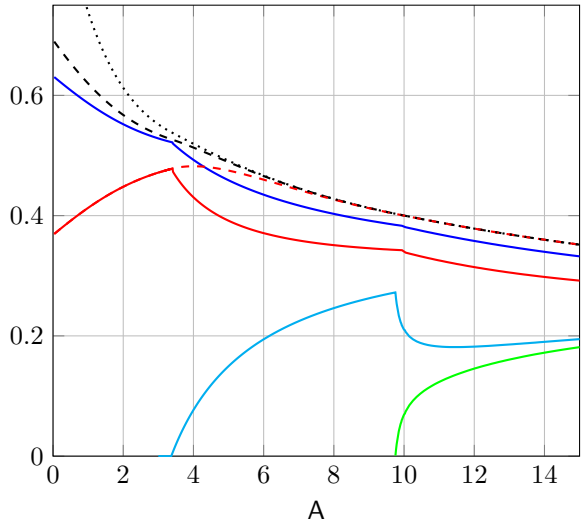}
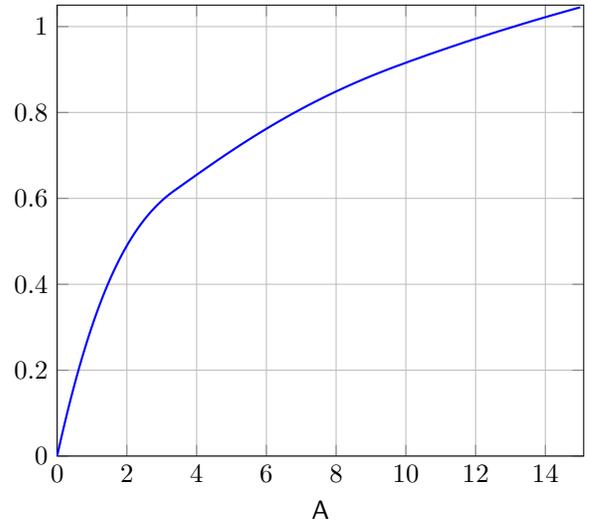
\caption{Plot of the probability values of the support points vs.~$\sfA$ where points of the support satisfy $0\le x_1 \le x_2 \le \sfA$.   The curves show:  $P_{X^\star}(0)$ (solid blue line); $P_{X^\star}(\sfA)$  (solid red line);  $P_{X^\star}(x_1)$ (solid cyan line);  $P_{X^\star}(x_2)$ (solid green line); upper bound in  \eqref{Eq:Bound_On_Probability_Values_poisson}  (dashed black line); upper bound  on $P_{X^\star}(\sfA)$ in  \eqref{Eq:Bound_On_Probability_Values_poisson_loc_dependent} (dashed red line); and upper bound  in  \eqref{eq:nostrongdataprocessingineq} with no strong data-processing term (dotted black line).   }
\label{fig:Probabilites_and_boound}
\end{subfigure}
~
\hspace{1cm}
\begin{subfigure}[t]{.4\linewidth}
\definecolor{mycolor1}{rgb}{0.00000,0.44700,0.74100}%
\begin{tikzpicture}

\begin{axis}[%
width=7cm,
height=6cm,
at={(1.011in,0.642in)},
scale only axis,
xmin=0,
xmax=15.1,
xlabel={$\sfA$},
ymin=0,
ymax=1.05,
axis background/.style={fill=white},
xmajorgrids,
ymajorgrids,
legend style={legend cell align=left, align=left, draw=white!15!black}
]
\addplot [color=blue,thick]
  table[row sep=crcr]{%
  0	0\\
0.0343434343434343	0.0125544841919615\\
0.0686868686868687	0.0249496729123573\\
0.103030303030303	0.0371859270805923\\
0.137373737373737	0.0492636426640856\\
0.171717171717172	0.0611832504629394\\
0.206060606060606	0.072945215864674\\
0.24040404040404	0.084550038569355\\
0.274747474747475	0.0959982522854861\\
0.309090909090909	0.10729042439709\\
0.343434343434343	0.118427155602457\\
0.377777777777778	0.129409079525065\\
0.412121212121212	0.140236862297263\\
0.446464646464646	0.150911202117306\\
0.480808080808081	0.161432828780409\\
0.515151515151515	0.171802503184514\\
0.549494949494949	0.182021016811511\\
0.583838383838384	0.192089191184679\\
0.618181818181818	0.202007877303165\\
0.652525252525253	0.211777955054353\\
0.686868686868687	0.221400332604977\\
0.721212121212121	0.230875945771921\\
0.755555555555556	0.240205757373605\\
0.78989898989899	0.249390756562943\\
0.824242424242424	0.258431958142845\\
0.858585858585859	0.267330401865266\\
0.892929292929293	0.276087151714831\\
0.927272727272727	0.284703295178066\\
0.961616161616162	0.293179942499294\\
0.995959595959596	0.301518225924243\\
1.03030303030303	0.30971929893246\\
1.06464646464646	0.317784335459578\\
1.0989898989899	0.325714529110537\\
1.13333333333333	0.333511092364831\\
1.16767676767677	0.341175255774858\\
1.2020202020202	0.34870826715845\\
1.23636363636364	0.356111390786652\\
1.27070707070707	0.363385906567815\\
1.30505050505051	0.370533109229041\\
1.33939393939394	0.377554307496036\\
1.37373737373737	0.384450823272376\\
1.40808080808081	0.391223990819211\\
1.44242424242424	0.397875155936383\\
1.47676767676768	0.40440567514593\\
1.51111111111111	0.410816914878929\\
1.54545454545455	0.417110250666596\\
1.57979797979798	0.423287066336547\\
1.61414141414141	0.42934875321509\\
1.64848484848485	0.435296709336392\\
1.68282828282828	0.441132338659349\\
1.71717171717172	0.446857050292936\\
1.75151515151515	0.452472257730796\\
1.78585858585859	0.457979378095801\\
1.82020202020202	0.463379831395269\\
1.85454545454545	0.468675039787508\\
1.88888888888889	0.473866426860309\\
1.92323232323232	0.478955416921972\\
1.95757575757576	0.483943434305452\\
1.99191919191919	0.488831902686119\\
2.02626262626263	0.493622244413642\\
2.06060606060606	0.498315879858448\\
2.09494949494949	0.502914226773174\\
2.12929292929293	0.507418699669505\\
2.16363636363636	0.511830709210747\\
2.1979797979798	0.51615166162045\\
2.23232323232323	0.520382958107378\\
2.26666666666667	0.524525994307061\\
2.3010101010101	0.528582159740169\\
2.33535353535354	0.53255283728787\\
2.36969696969697	0.536439402684357\\
2.4040404040404	0.540243224026638\\
2.43838383838384	0.543965661301712\\
2.47272727272727	0.547608065931173\\
2.50707070707071	0.551171780333303\\
2.54141414141414	0.554658137502639\\
2.57575757575758	0.558068460607015\\
2.61010101010101	0.561404062602029\\
2.64444444444444	0.564666245862872\\
2.67878787878788	0.56785630183342\\
2.71313131313131	0.570975510692485\\
2.74747474747475	0.574025141037079\\
2.78181818181818	0.577006449582549\\
2.81616161616162	0.57992068087939\\
2.85050505050505	0.582769067046558\\
2.88484848484848	0.585552827521056\\
2.91919191919192	0.58827316882358\\
2.95353535353535	0.590931284339967\\
2.98787878787879	0.593528354118199\\
3.02222222222222	0.59606554468068\\
3.05656565656566	0.598544008851516\\
3.09090909090909	0.600964885598484\\
3.12525252525253	0.603329299889409\\
3.15959595959596	0.605638362562609\\
3.19393939393939	0.607893170211105\\
3.22828282828283	0.610094805080252\\
3.26262626262626	0.612244334978458\\
3.2969696969697	0.614342813200654\\
3.33131313131313	0.616391278464157\\
3.36565656565657	0.618390754856579\\
3.4	0.62034225179543\\
%3	0.594430610901382\\
%3.05263157894737	0.598263039952178\\
%3.10526315789474	0.601959911883644\\
%3.15789473684211	0.605525267935662\\
%3.21052631578947	0.608963071153005\\
%3.26315789473684	0.612277205483695\\
%3.31578947368421	0.615471475047094\\
%3.36842105263158	0.61854960787881\\
3.42105263157895	0.621574619298788\\
3.47368421052632	0.62460408682152\\
3.52631578947368	0.62763679939487\\
3.57894736842105	0.630671618135925\\
3.63157894736842	0.633707472852837\\
3.68421052631579	0.636743358114095\\
3.73684210526316	0.639778329489677\\
3.78947368421053	0.642811499959296\\
3.84210526315789	0.645842036484976\\
3.89473684210526	0.648869156745762\\
3.94736842105263	0.651892126032699\\
4	0.654910254302241\\
4.1	0.660629071257674\\
4.14871794871795	0.663406754762156\\
4.1974358974359	0.666178340296656\\
4.24615384615385	0.668943421687397\\
4.29487179487179	0.671701615274627\\
4.34358974358974	0.674452558583944\\
4.39230769230769	0.677195909080287\\
4.44102564102564	0.679931343001307\\
4.48974358974359	0.682658554266618\\
4.53846153846154	0.68537725345927\\
4.58717948717949	0.688087166875698\\
4.63589743589744	0.690788035640284\\
4.68461538461538	0.693479614880702\\
4.73333333333333	0.696161672960167\\
4.78205128205128	0.698833990762816\\
4.83076923076923	0.70149636102851\\
4.87948717948718	0.704148587733465\\
4.92820512820513	0.706790485513268\\
4.97692307692308	0.709421879125008\\
5.02564102564103	0.712042602945394\\
5.07435897435897	0.714652500501997\\
5.12307692307692	0.71725142403485\\
5.17179487179487	0.719839234085946\\
5.22051282051282	0.722415799114292\\
5.26923076923077	0.724980995134416\\
5.31794871794872	0.727534705376404\\
5.36666666666667	0.730076819965737\\
5.41538461538462	0.732607235621323\\
5.46410256410256	0.735125855370352\\
5.51282051282051	0.737632588278671\\
5.56153846153846	0.740127349195568\\
5.61025641025641	0.742610058511938\\
5.65897435897436	0.745080641930906\\
5.70769230769231	0.747539030250136\\
5.75641025641026	0.749985159155026\\
5.8051282051282	0.752418969022166\\
5.85384615384615	0.754840404732453\\
5.9025641025641	0.757249415493264\\
5.95128205128205	0.759645954669228\\
6	0.762029979621073\\
6.11557788944724	0.767637280998367\\
6.2713567839196	0.775079961973713\\
6.42713567839196	0.782392773568015\\
6.58291457286432	0.789575262092383\\
6.73869346733668	0.796627253213853\\
6.89447236180905	0.803548821586225\\
7.05025125628141	0.810340263947354\\
7.20603015075377	0.817002075084587\\
7.36180904522613	0.823534926230899\\
7.51758793969849	0.829939645600877\\
7.67336683417085	0.836217200863585\\
7.82914572864322	0.842368683389873\\
7.98492462311558	0.848395294125637\\
8.14070351758794	0.854298330919391\\
8.2964824120603	0.860079177119965\\
8.45226130653266	0.865739291257573\\
8.60804020100502	0.871280197630848\\
8.76381909547739	0.876703477661874\\
8.91959798994975	0.882010761922579\\
9	0.884781728659965\\
9.1045	0.88812540624133\\
9.2091	0.891488282850996\\
9.3136	0.894846063975276\\
9.4182	0.898194184781119\\
9.5227	0.901494604995558\\
9.6273	0.904707031029915\\
9.7318	0.907868494956731\\
9.752	0.908474135391212\\
9.755	0.90856393368696\\
9.76	0.908713557426875\\
9.765	0.908863139392592\\
9.77	0.909012679064479\\
9.79	0.909610418546407\\
9.85	0.911274107068111\\
9.9	0.91288600431212\\
9.92	0.913481669596017\\
9.937	0.913920615738572\\
9.945	0.914134143092975\\
9.95	0.9143682553399\\
9.955	0.914516253580029\\
9.965	0.914812126089211\\
9.99	0.915464440109554\\
10	0.915923605461455\\
10.025	0.916580517677273\\
10.03	0.916805782642913\\
10.0455	0.917237827091861\\
10.15	0.920381806841576\\
10.2545	0.923285199152007\\
10.3591	0.926229485990104\\
10.4636	0.929439319248056\\
10.5682	0.932428104651601\\
10.6727	0.935466135534588\\
10.7773	0.938264322753206\\
10.8818	0.941334084189751\\
10.9864	0.944182134553018\\
11.0909	0.947074243436847\\
11.1955	0.949977607925751\\
11.3	0.952920193967189\\
11.4	0.955595851887975\\
11.5	0.958333968485789\\
11.6	0.961093630809974\\
11.7	0.963866645340828\\
11.8	0.966383207559083\\
11.9	0.969199095059503\\
12	0.971783789674196\\
12.1	0.974367353119987\\
12.2	0.976981807526348\\
12.3	0.979619459193345\\
12.4	0.982280719041044\\
12.5	0.984939804175918\\
12.6	0.987356655435422\\
12.7	0.990058156876808\\
12.8	0.992516683259354\\
12.9	0.994992065632089\\
13	0.997492437594037\\
13.1053	1.00001599669174\\
13.2105	1.00281314419144\\
13.3158	1.00537496526126\\
13.4211	1.00789852600142\\
13.5263	1.01051230651864\\
13.6316	1.0130998965014\\
13.7368421052632	1.01559051211973\\
13.8421052631579	1.01811237323968\\
13.9473684210526	1.02061926919419\\
14.0526315789474	1.02311123804826\\
14.1578947368421	1.02558833473102\\
14.2631578947368	1.02805061410964\\
14.3684210526316	1.03049813101142\\
14.4736842105263	1.03293094025887\\
14.5789473684211	1.03534909670134\\
14.6842105263158	1.0377526552434\\
14.7894736842105	1.04014167087016\\
14.8947368421053	1.04251619866975\\
15	1.0448762938533\\
};

\end{axis}

\end{tikzpicture}%
\caption{Plot of $C(\sfA)$ vs. $\sfA$.}
\label{fig:capacity}
\end{subfigure}
\caption{Examples of the optimal input distributions and capacity  vs. $\sfA$. The  probabilities and the capacity for  $\sfA>\bar{\sfA}\approx 3.4$  were computed numerically, and the probabilities and the capacity for $\sfA \le \bar{\sfA}$ are given in \eqref{eq:OptimazingDistributionSmallAmplitude} and \eqref{eq:ShamaiCapacity}, respectively.  }
\label{fig:Plot_of_probabilitiues}
\end{figure}

\subsection{On the Order of the Bounds} 
 First, note that almost all of our bounds  depend on  the value of  $C(\sfA)$, which is currently unknown for $\sfA \ge \bar{\sfA}\approx 3.4$.  However, this not a limitation of our result,
 as we do have access to upper  and lower bounds on $C(\sfA)$   that are tight for large $\sfA$ such as those in \cite{lapidoth2009capacity}, which suggest that
 \begin{equation}
C(\sfA)= \frac{1}{2} \log(\sfA) -\frac{1}{2} \log \left(  \frac{\pi \rme}{2}\right)+ o_{\sfA}(1).  \label{eq:Moser_Lapidos_bound}
\end{equation} 
Moreover, some of the bounds in Theorem~\ref{thm:Main_Result}, such as those in \eqref{eq:lower_bound_on_P_X(A)} and \eqref{eq:Bounds_on_the_Number_Of_Poitns}, while are firm, are meant to be used for large values of $\sfA$. Therefore, combing the bounds in Theorem~\ref{thm:Main_Result} with the bound in \eqref{eq:Moser_Lapidos_bound}, we arrive at the following: 
\begin{align}
 \Omega \left(\frac{\rme^{2\sfA}}{ \sfA^{2\sfA \rme  \log(\sfA)+2}   }  \right)  \le P_{X^\star} (x) & \le O \left(\frac{1}{\sqrt{\sfA}} \right) , \qquad \,  x\in \supp(P_{X^\star}),\\
   \Omega( \sqrt{\sfA} )\le   | \supp(P_{X^\star}) | &\le  O \left(  \sfA  \log^2(\sfA) \right). 
\end{align} 
Thus, the order of the lower bound on the number of points  is $\sqrt{\sfA}$, and the order of the upper bound  on the number of points  is $\sfA  \log^2(\sfA) $.  It is interesting to speculate as to the reason why the bounds do not match and  have different orders.

  First, note that Theorem~\ref{thm:Main_Result} presents two upper bounds on the number of points. The first and implicit bound in \eqref{eq:First_implicit_bound_on_point} depends on the value of $P_{X^\star}(\sfA)$.  The second bound in \eqref{eq:Bounds_on_the_Number_Of_Poitns} is an explicit bound in terms of $\sfA$ and is derived by plugging in the lower bound on $P_{X^\star}(\sfA)$ in \eqref{eq:lower_bound_on_P_X(A)} into the first bound in  \eqref{eq:First_implicit_bound_on_point}.  We suspect that one of the reason why the  bounds do not match  is due to the lower bound on   $P_{X^\star}(\sfA)$ in \eqref{eq:lower_bound_on_P_X(A)}, which we think is not tight.  Hence, one interesting future direction is to improve the lower bound on  the value of $P_{X^\star}(\sfA)$ which, in view of \eqref{eq:First_implicit_bound_on_point}, would lead to a better upper bound on $ | \supp(P_{X^\star}) |$. However, to the best of our knowledge, there are no other methods for finding lower or upper bounds on the probabilities of the optimal input distribution. Indeed, one of the contributions of this work is the introduction of two such methods, one for finding an upper bound  on the values of probabilities and one for finding a lower bound. 
  
Finally, we would like to point out that numerical simulations are not useful for predicting the order of the number of points. For large values of $\sfA$, the simulations become numerically unstable, and it is difficult to calculate the optimal input distribution and predict the order of the number of points. 
 
 \subsection{On the Upper Bounds on the Probabilities in  \eqref{Eq:Bound_On_Probability_Values_poisson} and \eqref{Eq:Bound_On_Probability_Values_poisson_loc_dependent}}
It is also interesting to ask how tight the upper bounds on the probabilities  in  \eqref{Eq:Bound_On_Probability_Values_poisson} and \eqref{Eq:Bound_On_Probability_Values_poisson_loc_dependent} are.  

Note that the bound in \eqref{Eq:Bound_On_Probability_Values_poisson} is universal and does not depend on the  positions of the probability masses, while the bound in \eqref{Eq:Bound_On_Probability_Values_poisson_loc_dependent} depends on the position of the points. 
The advantage of the bound in  \eqref{Eq:Bound_On_Probability_Values_poisson_loc_dependent} is that it can be tighter than the universal bound in  \eqref{Eq:Bound_On_Probability_Values_poisson}.    For example, in the regime where $\sfA \le \bar{\sfA}$ where we only have two points in the support, the bound is achieved with equality.  The clear disadvantage of the bound in \eqref{Eq:Bound_On_Probability_Values_poisson_loc_dependent} is that we do not know the location of the points (except $0$ and $\sfA$). However,  such a bound might become useful once  better estimates for the locations of the mass  points are found.  Some preliminary estimates of the locations  are provided in \eqref{eq:Bound_on_other_points}.

Fig.~\ref{fig:Probabilites_and_boound} plots the upper bounds in \eqref{Eq:Bound_On_Probability_Values_poisson} and \eqref{Eq:Bound_On_Probability_Values_poisson_loc_dependent} and compares them to the values of $P_{X^\star}$.  To create the plot, in the regime  $\sfA \le \bar{\sfA}$, we have used the exact expressions for $P_{X^\star}(\sfA)$ and $P_{X^\star}(0)$ in \eqref{eq:OptimazingDistributionSmallAmplitude}.  To create the plot in the regime $\sfA >\bar{\sfA}$, first note that the upper bounds in \eqref{Eq:Bound_On_Probability_Values_poisson} can be loosened to 
\begin{align}
P_{X^\star}(x)& \le \rme^{- \frac{I(\widetilde{X};\widetilde{Y})}{ 1-\rme^{-\sfA}}  }, \quad x\in \supp(P_{X^\star}),\\
P_{X^\star}(x) &\le \rme^{-I(\widetilde{X};\widetilde{Y})-x \frac{\rme^{-x}}{1-\rme^{-x}}}, \quad x\in \supp(P_{X^\star}) \setminus \{0\}, \label{eq:Second_boud_points_temp}
\end{align} 
where we have used the fact that $C(\sfA) \ge I(\widetilde{X};\widetilde{Y})$ for any random variable $\widetilde{X} \in [0,\sfA]$ and where $\widetilde{Y}$ is induced by $\widetilde{X}$.  Therefore, since we can choose any $\widetilde{X} \in [0,\sfA]$, we selected it to be the one that is the output of the numerical simulation.  The bound in \eqref{eq:Second_boud_points_temp} is only computed for $x=\sfA$ as we do not know the locations of other points and only have estimates for these. 
%Furthermore, to compute  the term $x \frac{\rme^{-x}}{1-\rme^{-x}}$  in  \eqref{eq:Second_boud_points_temp}, we also use the value of $x$ computed by the simulations.  We, however, warn that while using the distribution found numerically to substitute $C(\sfA)$ by  $I(\widetilde{X};Y)$ results in a valid upper bound, substituting $x\in \supp(P_{X^\star}) $ by a numerically computed point  does not result in a valid bound. Therefore, the latter plot should only be taught of as  approximation of the real bounds. 
From the simulations in Fig.~\ref{fig:Probabilites_and_boound}, the bounds in \eqref{Eq:Bound_On_Probability_Values_poisson} and \eqref{Eq:Bound_On_Probability_Values_poisson_loc_dependent} appear to be relatively tight. 

The bound in \eqref{Eq:Bound_On_Probability_Values_poisson} relies on the strong data-processing inequality.   Specifically, the factor $\frac{1}{1-\rme^{-\sfA}}$ in the exponent comes from using  the strong data-processing inequality.  The dotted black curve in Fig.~\ref{fig:Probabilites_and_boound} plots the loosened version of the bound in \eqref{Eq:Bound_On_Probability_Values_poisson} that ignores the contribution of the strong data-processing inequality, that is we plot
\begin{equation}\label{eq:nostrongdataprocessingineq}
P_{X^\star}(x) \le \rme^{- I(\widetilde{X};\widetilde{Y})},  \quad x\in \supp(P_{X^\star}).
\end{equation} 
From the comparison Fig.~\ref{fig:Probabilites_and_boound}, we see that the contribution of the strong data-processing inequality is nontrivial, especially for small and medium values of $\sfA$.

\subsection{On the Bound in \eqref{eq:Bound_on_other_points} } 
In addition to finding bounds on  the number of points and the values of the probabilities,   we have also provided additional information about the location of the points. Specifically,  \eqref{eq:Bound_on_other_points}  provides  information about the location of support points other than $0$ and $\sfA$.   

From the bound in \eqref{eq:Bound_on_other_points}, we see that the second-largest point can  never be too close to $\sfA$. Specifically,  according to the bound in \eqref{eq:Bound_on_other_points},  the gap between $\sfA$  (i.e., the largest point) and the second-largest point is at least one.  In fact, the numerical simulations shown in Fig.~\ref{fig:Locations} suggest that this gap is  much larger.  In particular, the simulations suggest that the gap is not constant but is an increasing function of $\sfA$.  Therefore, one interesting future direction would be to verify this behavior and produce a better bound in~\eqref{eq:Bound_on_other_points} than $\sfA-1$.

Similarly, from the lower bound in \eqref{eq:Bound_on_other_points},  we see that the second smallest point cannot be too close to the  zero point. However,  as $\sfA$ increases, the distance is allowed to get smaller.  Note that  our limited simulation results suggest a better lower bound, namely $x^\star\ge1$.  Therefore, one interesting future direction would be to either demonstrate the existence of a mass point in the range $(0,1)$ or show that there is no such mass point. Note that the work of \cite{mceliece1979practical} already showed that there is at most one point in the range $(0,1)$. 

Beyond theoretical interest, the existence of the estimates for the mass points' location might also be of interest from the practical point of view. 
 As the existence of such estimates can also impact the design of practical constellations for the Poisson noise channel.

 \subsection{On the Equivocation, Symbol  Error Probability, and Entropy} 
 
 It is well-known that the capacity-achieving distribution should be `difficult' to detect or estimate on the per-symbol basis.  To  make this statement explicit, we consider the equivocation   $H(X^\star|Y^\star)$  and the probability of error under the maximum a  posteriori (MAP) rule (i.e., $P_e=\mathbb{P}[X^\star \neq \hat{X}(Y^\star)]$ where  $\hat{X}(Y^\star)$ is the MAP decoder). The plots in Fig.~\ref{fig:H(X|Y)} and Fig.~\ref{fig:P_e_Plot} show that the equivocation and $P_e$ for the capacity-achieving input have relatively high values.  
With Theorem~\ref{thm:Main_Result} at our disposal, we can now show the following result regarding the asymptotic behavior of  the error probability. 
\begin{prop} Let $P_e=\mathbb{P}[X^\star \neq \hat{X}(Y^\star)]$ where  $\hat{X}(y) = \max_{x \in \supp(P_{X^\star}) } P_{X^\star|Y^\star}(x|y) $. Then, 
\begin{align}
\liminf_{\sfA \to \infty} P_e & \ge 1-  \sqrt{\frac{2 }{ \pi}}, \label{eq:LowerBound_P-e}\\
\liminf_{\sfA \to \infty}  H(X^\star|Y^\star)  &\ge 2 \left(1-  \sqrt{\frac{2 }{ \pi}}\right).  \label{eq:LowerBound_HXY}
\end{align}
\end{prop} 
\begin{proof}
The probability of error for the MAP  rule can be written as
\begin{align}
P_e&=1- \bbE \left[ \max_{x \in \supp(P_{X^\star}) } P_{X^\star|Y^\star}(x|Y^\star) \right]  \label{eq:Pe_expression_Sason_Verdu}\\
&=1- \bbE \left[ \max_{x \in \supp(P_{X^\star}) } \frac{P_{Y|X}(Y^\star|x) P_{X^\star}(x)}{P_{Y^\star}(Y^\star)} \right ]\\
& \ge 1-  \rme^{- \frac{C(\sfA)}{ 1-\rme^{-\sfA}}  }\: \bbE \left[ \max_{x \in \supp(P_{X^\star}) } \frac{P_{Y|X}(Y^\star|x) }{P_{Y^\star}(Y^\star)} \right] \label{eq:P_e_bound_Using_Bound_on_Probabilities}\\
& = 1-  \rme^{- \frac{C(\sfA)}{ 1-\rme^{-\sfA}}  }  \sum_{y=0}^\infty \max_{x \in \supp(P_{X^\star}) } P_{Y|X}(y|x) \\
& = 1-    \rme^{- \frac{C(\sfA)}{ 1-\rme^{-\sfA}}  }  \left( \sum_{y<\sfA} \max_{x \in \supp(P_{X^\star}) } P_{Y|X}(y|x)   +\sum_{y \ge \sfA} \max_{x \in \supp(P_{X^\star}) } P_{Y|X}(y|x)  \right)\\
& \ge 1-  \rme^{- \frac{C(\sfA)}{ 1-\rme^{-\sfA}}  }   \left( \sum_{y<\sfA}  P_{Y|X}(y|y)+\sum_{y \ge \sfA}  P_{Y|X}(y|\sfA) \right) \label{eq:P_e_bound_max_of_P_{Y|X}}\\
& = 1-  \rme^{- \frac{C(\sfA)}{ 1-\rme^{-\sfA}}  }   \left( 1+ \sum_{1 \le y<\sfA}  \frac{y^y e^{-y}}{y!} + \mathbb{P}[Y \ge \sfA|X=\sfA] \right)\\
& \ge  1-  \rme^{- \frac{C(\sfA)}{ 1-\rme^{-\sfA}}  }   \left( 1+ \sum_{1 \le y<\sfA}  \frac{1}{\sqrt{2 \pi y}} + \mathbb{P}[Y \ge \sfA|X=\sfA] \right) \label{eq:P_e_bound_Using_Stirlings_bound}\\
& \ge 1- \rme^{- \frac{C(\sfA)}{ 1-\rme^{-\sfA}}  }    \left( 1+ \sqrt{\frac{2 \sfA}{ \pi}}+\frac{1}{ \sqrt{2\pi \sfA}} +c + \mathbb{P}\left[Y \ge \sfA \middle| X=\sfA \right] \right), \label{eq:Harmonic_Numer_bounds}
\end{align}
where the expression for the probability of error in \eqref{eq:Pe_expression_Sason_Verdu} comes from \cite{sason2017arimoto}; \eqref{eq:P_e_bound_Using_Bound_on_Probabilities} follows by using the bound in \eqref{Eq:Bound_On_Probability_Values_poisson}; \eqref{eq:P_e_bound_max_of_P_{Y|X}} follows by using that $\max_{x \in [0,\sfA]}   P_{Y|X}(y|x) \le  P_{Y|X}(y|y)$ for $y < \sfA$ and  $\max_{x \in [0,\sfA]}   P_{Y|X}(y|x) \le  P_{Y|X}(y|\sfA)$ for $y  \ge  \sfA$;  \eqref{eq:P_e_bound_Using_Stirlings_bound} follows by using Stirling's bound  $\frac{y^y e^{-y}}{y!} \le \frac{1}{ \sqrt{2 \pi y}}, \, y \ge 1$; and~\eqref{eq:Harmonic_Numer_bounds} follows by using the bound $\sum_{1 \le y<\sfA}  \frac{1}{\sqrt{y}} \le 2 \sqrt{y}+y^{-\frac{1}{2}}+c$ where $c>0$ is some universal constant \cite{apostol1998introduction}. 

To conclude the proof of \eqref{eq:LowerBound_P-e},  we use the bound in \eqref{eq:Harmonic_Numer_bounds}  together with the asymptotic expression for $C(\sfA)$ in \eqref{eq:Moser_Lapidos_bound}. 
The proof of \eqref{eq:LowerBound_HXY} follows by using the following bound  \cite{baladova1966minimum}: 
\begin{equation}
H(X|Y) \ge 2 P_e. 
\end{equation} 
\end{proof}

\begin{rem}

The bound in  \eqref{eq:LowerBound_HXY} can be improved, albeit at the expense of a more complicated expression. In particular, the bound relies on the inequality $H(X|Y) \ge 2P_e$. There are several stronger versions of this inequality; the interested reader is referred to \cite{sason2017arimoto} for a summary of these inequalities. 
\end{rem}

The entropy of the optimal input distribution  vs.~$\sfA$ is plotted in Fig.~\ref{fig:Entropy}.   We observe a particular behavior of the entropy in the simulated range of $\sfA$: the rate of increase has finite jumps approximately at the levels $\log(k)$ of entropy, where the cardinality of the optimal input distribution increases from $k$ to $k+1$ points. These levels correspond to an approximate uniform input distribution on the $k$ amplitude levels: this behavior is also confirmed by the mass probabilities plotted in Fig.~\ref{fig:Probabilites_and_boound}. When the rate of increase of the entropy is not compensated by a sufficiently large rate of decrease of the equivocation (see Fig.~\ref{fig:H(X|Y)}), the rate of increase of capacity must be sustained by boosting the input entropy: This is done by increasing the cardinality of the input distribution. It is interesting to understand how the rate of increase of capacity is split between entropy and equivocation: If one could prove that the equivocation is upper-bounded by a constant, then this would show that the equivocation does not provide degrees of freedom to channel capacity for large $\sfA$, and thus the whole rate should be sustained by the input entropy by increasing the cardinality of $\supp(P_{X^\star})$. This hypothesis would imply $|\supp(P_{X^\star})| \approx \sqrt{\sfA}$ for large $\sfA$.

\begin{figure}
\center

%\begin{subfigure}[t]{.4\linewidth}
%\input{FIG1/Positions_log.tex}
%\caption{Plot of the locations of the support points of $P_{X^\star}$ vs. $\sfA$ where the log scale is used for the vertical axis.}
%\label{fig:Locations_Log_scale}
%\end{subfigure}

\begin{subfigure}[t]{.4\linewidth}
\definecolor{mycolor1}{rgb}{0.00000,0.44700,0.74100}%
\begin{tikzpicture}

\begin{axis}[%
width=7cm,
height=6cm,
at={(1.011in,0.642in)},
scale only axis,
xmin=0,
xmax=18.1,
xlabel style={font=\color{white!15!black}},
xlabel={$\sfA$},
ymin=0.05,
ymax=0.55,
ylabel style={font=\color{white!15!black}},
ylabel={$H(X^*|Y^*)$},
axis background/.style={fill=white},
xmajorgrids,
ymajorgrids,
legend style={legend cell align=left, align=left, draw=white!15!black}
]
\addplot [color=blue,thick]
  table[row sep=crcr]{%
0.377777777777778	0.537432353443274\\
0.755555555555556	0.433887071830777\\
1.13333333333333	0.346212245400557\\
1.51111111111111	0.273130226595039\\
1.88888888888889	0.213143119650757\\
2.26666666666667	0.164631509706577\\
2.64444444444444	0.125950719202421\\
3	0.0971029249094638\\
3.05263157894737	0.0933769446658873\\
3.10526315789474	0.0897800833833398\\
3.15789473684211	0.0863086406966694\\
3.21052631578947	0.082958980286725\\
3.26315789473684	0.0797275310620276\\
3.31578947368421	0.0766107881909628\\
3.36842105263158	0.0736037904064378\\
3.42105263157895	0.111498675650296\\
3.47368421052632	0.13680754600138\\
3.52631578947368	0.157232163262063\\
3.57894736842105	0.174425530247318\\
3.63157894736842	0.189211963979847\\
3.68421052631579	0.202096135758312\\
3.73684210526316	0.213422024113454\\
3.78947368421053	0.22344018583169\\
3.84210526315789	0.232341551348572\\
3.89473684210526	0.240276432138401\\
3.94736842105263	0.247366108636111\\
4	0.253710336256024\\
4.1	0.263998707931128\\
4.14871794871795	0.268282994645116\\
4.1974358974359	0.272149321451534\\
4.24615384615385	0.275634590250871\\
4.29487179487179	0.278771383405236\\
4.34358974358974	0.281588623461423\\
4.39230769230769	0.28411210752734\\
4.44102564102564	0.286364944776504\\
4.48974358974359	0.288367918063434\\
4.53846153846154	0.290139785363962\\
4.58717948717949	0.291697532980831\\
4.63589743589744	0.293056589707024\\
4.68461538461538	0.294231009108039\\
4.73333333333333	0.295233625562257\\
4.78205128205128	0.29607618854375\\
4.83076923076923	0.296769478745462\\
4.87948717948718	0.297323408953285\\
4.92820512820513	0.297747112043203\\
4.97692307692308	0.29804901804816\\
5.02564102564103	0.298236921902384\\
5.07435897435897	0.298318043198701\\
5.12307692307692	0.298299079074282\\
5.17179487179487	0.298186251161112\\
5.22051282051282	0.297985347390747\\
5.26923076923077	0.297701759322133\\
5.31794871794872	0.297340515561135\\
5.36666666666667	0.296906311757287\\
5.41538461538462	0.296403537593681\\
5.46410256410256	0.295836301127491\\
5.51282051282051	0.29520845078958\\
5.56153846153846	0.294523595309917\\
5.61025641025641	0.293785121800365\\
5.65897435897436	0.292996212196327\\
5.70769230769231	0.292159858233132\\
5.75641025641026	0.291278875111074\\
5.8051282051282	0.290355913984266\\
5.85384615384615	0.289393473392117\\
5.9025641025641	0.288393909738459\\
5.95128205128205	0.287359446911022\\
6	0.286292185123673\\
6.11557788944724	0.283627267710297\\
6.2713567839196	0.27982244741629\\
6.42713567839196	0.275796269574819\\
6.58291457286432	0.271591798621647\\
6.73869346733668	0.267245531043759\\
6.89447236180905	0.262788517771597\\
7.05025125628141	0.258247264431799\\
7.20603015075377	0.253644459998786\\
7.36180904522613	0.248999571532421\\
7.51758793969849	0.244329333212627\\
7.67336683417085	0.239648151011236\\
7.82914572864322	0.234968439338767\\
7.98492462311558	0.230300902345517\\
8.14070351758794	0.225654769953358\\
8.2964824120603	0.221037996727951\\
8.45226130653266	0.216457430203114\\
8.60804020100502	0.211918954086837\\
8.76381909547739	0.207427610756166\\
8.91959798994975	0.202987706600543\\
9	0.200718167500728\\
9.1045	0.197793449147779\\
9.2091	0.194878058576402\\
9.3136	0.192005544098693\\
9.4182	0.189168985944413\\
9.5227	0.186345885104605\\
9.6273	0.183554376707498\\
9.7318	0.180795057600753\\
9.752	0.180265114800374\\
9.755	0.186376432083928\\
9.76	0.201462521142822\\
9.765	0.213410572590054\\
9.77	0.223679791783632\\
9.79	0.251352702554589\\
9.85	0.294789648268221\\
9.9	0.312648111136672\\
9.92	0.316900345712881\\
9.937	0.32165614351444\\
9.945	0.322538558881021\\
9.95	0.323669888675554\\
9.955	0.324526994997502\\
9.965	0.326157853527545\\
9.99	0.330024888680164\\
10	0.33113145863495\\
10.025	0.332849277477727\\
10.03	0.333775715171045\\
10.0455	0.336044549050555\\
10.15	0.344322060537033\\
10.2545	0.349092282237055\\
10.3591	0.352182254094863\\
10.4636	0.354279510419702\\
10.5682	0.355662159913069\\
10.6727	0.356498465900498\\
10.7773	0.356907536284063\\
10.8818	0.35722678230814\\
10.9864	0.357158500690765\\
11.0909	0.356889222524595\\
11.1955	0.356580524702522\\
11.3	0.356042557245211\\
11.4	0.355502070232395\\
11.5	0.354853345422582\\
11.6	0.354060015819026\\
11.7	0.353248266510368\\
11.8	0.352396166368191\\
11.9	0.351488691686627\\
12	0.350473995362393\\
12.1	0.349424020765704\\
12.2	0.348333682527325\\
12.3	0.347216524933202\\
12.4	0.346018487061506\\
12.5	0.344881378169409\\
12.6	0.343611983919757\\
12.7	0.342419065345844\\
12.8	0.341073506153503\\
12.9	0.339791414187781\\
13	0.338440940655898\\
13.1053	0.336956635118002\\
13.2105	0.335560773870686\\
13.3158	0.334074912279833\\
13.4211	0.332612830523618\\
13.5263	0.331069692139602\\
13.6316	0.329572574518937\\
13.7368421052632	0.328019813491933\\
13.8421052631579	0.32646637266027\\
13.9473684210526	0.324898620712211\\
14.0526315789474	0.32331880147547\\
14.1578947368421	0.321727894971027\\
14.2631578947368	0.320126825471994\\
14.3684210526316	0.318516466580601\\
14.4736842105263	0.316897644918266\\
14.5789473684211	0.315271143437563\\
14.6842105263158	0.313637704406657\\
14.7894736842105	0.311998032108491\\
14.8947368421053	0.310352795289925\\
15	0.308702629391614\\
15.5	0.300814786341096\\
16	0.292884750860906\\
16.5	0.284957586571362\\
17	0.277069356301693\\
17.5	0.269249010821099\\
18	0.261519790270362\\
18.1	0.259986590913357\\
18.12	0.259680490552468\\
18.13	0.259527508653385\\
18.135	0.259451034828787\\
18.136	0.259435741432521\\
};
%\addlegendentry{data1}

\end{axis}
\end{tikzpicture}%
\caption{Equivocation $H(X^\star|Y^\star)$ vs. $\sfA$.}
\label{fig:H(X|Y)}
\end{subfigure}
~
\hspace{1.4cm}
\begin{subfigure}[t]{.4\linewidth}
\definecolor{mycolor1}{rgb}{0.00000,0.44700,0.74100}%
\begin{tikzpicture}

\begin{axis}[%
width=7cm,
height=6cm,
at={(1.011in,0.642in)},
scale only axis,
xmin=0,
xmax=18.1,
xlabel style={font=\color{white!15!black}},
xlabel={$\sfA$},
ymin=0,
ymax=0.3,
ylabel style={font=\color{white!15!black}},
ylabel={$P_e$},
axis background/.style={fill=white},
xmajorgrids,
ymajorgrids,
legend style={legend cell align=left, align=left, draw=white!15!black}
]
\addplot [color=blue, thick]
  table[row sep=crcr]{%
0.377777777777778	0.264434213299954\\
0.755555555555556	0.18917003775631\\
1.13333333333333	0.134661466399557\\
1.51111111111111	0.0953890865646213\\
1.88888888888889	0.0672502471631763\\
2.26666666666667	0.0472019436592879\\
2.64444444444444	0.0329965252956417\\
3	0.0234797319455637\\
3.05263157894737	0.0223208875198573\\
3.10526315789474	0.0212179690705541\\
3.15789473684211	0.020168367865371\\
3.21052631578947	0.0191695908670964\\
3.26315789473684	0.0182192559616381\\
3.31578947368421	0.0173150873650068\\
3.36842105263158	0.0164543225022948\\
3.42105263157895	0.0234383622260754\\
3.47368421052632	0.0300759011076643\\
3.52631578947368	0.036457466950774\\
3.57894736842105	0.0425992762083183\\
3.63157894736842	0.0485161264684912\\
3.68421052631579	0.0542215471623891\\
3.73684210526316	0.059727931167635\\
3.78947368421053	0.0650466505502805\\
3.84210526315789	0.0701881587049273\\
3.89473684210526	0.075162080778499\\
3.94736842105263	0.0799772939560945\\
4	0.0846419989351371\\
4.1	0.0924909974411163\\
4.14871794871795	0.0938453791582812\\
4.1974358974359	0.0951989585104102\\
4.24615384615385	0.0965517566588222\\
4.29487179487179	0.0979037676753671\\
4.34358974358974	0.0992549614757611\\
4.39230769230769	0.10060528661478\\
4.44102564102564	0.101954672929588\\
4.48974358974359	0.103303034020764\\
4.53846153846154	0.104650269563743\\
4.58717948717949	0.10599626744654\\
4.63589743589744	0.107340905732459\\
4.68461538461538	0.108684054449174\\
4.73333333333333	0.110025577207864\\
4.78205128205128	0.111365332658106\\
4.83076923076923	0.112703175785895\\
4.87948717948718	0.113022367947153\\
4.92820512820513	0.112385669824533\\
4.97692307692308	0.111795561194915\\
5.02564102564103	0.11125079077286\\
5.07435897435897	0.110750138683141\\
5.12307692307692	0.110292415241836\\
5.17179487179487	0.109876459851469\\
5.22051282051282	0.109501139994507\\
5.26923076923077	0.109165350312055\\
5.31794871794872	0.108868011756691\\
5.36666666666667	0.108608070810169\\
5.41538461538462	0.108384498758152\\
5.46410256410256	0.10819629101531\\
5.51282051282051	0.108042466495102\\
5.56153846153846	0.107922067019289\\
5.61025641025641	0.107834156762896\\
5.65897435897436	0.107777821730777\\
5.70769230769231	0.107752169262371\\
5.75641025641026	0.107756327561562\\
5.8051282051282	0.107789445248802\\
5.85384615384615	0.107850690932902\\
5.9025641025641	0.107939252800119\\
5.95128205128205	0.10805433821832\\
6	0.108195173354227\\
6.11557788944724	0.106515397672548\\
6.2713567839196	0.103130159354142\\
6.42713567839196	0.100170596652132\\
6.58291457286432	0.0976107816215154\\
6.73869346733668	0.0954257690333188\\
6.89447236180905	0.0935916362885072\\
7.05025125628141	0.0920855113215663\\
7.20603015075377	0.0908855888390409\\
7.36180904522613	0.0899711361024462\\
7.51758793969849	0.0893224896175134\\
7.67336683417085	0.0862955078952846\\
7.82914572864322	0.0832303662094631\\
7.98492462311558	0.0804889241335851\\
8.14070351758794	0.0780542727848014\\
8.2964824120603	0.0759099463030531\\
8.45226130653266	0.0740399784332811\\
8.60804020100502	0.0724289467033775\\
8.76381909547739	0.0710620054130491\\
8.91959798994975	0.0699249087765061\\
9	0.0694388459765884\\
9.1045	0.0688724353320046\\
9.2091	0.0683898267912628\\
9.3136	0.0666385479753181\\
9.4182	0.0649220381016314\\
9.5227	0.0633152056682676\\
9.6273	0.0618134234544477\\
9.7318	0.0604165632658519\\
9.752	0.0601581822298974\\
9.755	0.0609265836233668\\
9.76	0.0636519579564573\\
9.765	0.0663104393244308\\
9.77	0.0689012119111124\\
9.79	0.077256735109694\\
9.85	0.0952094277677484\\
9.9	0.105057413715267\\
9.92	0.107718304586647\\
9.937	0.110879875110295\\
9.945	0.111499277126846\\
9.95	0.112286373602133\\
9.955	0.112897560637427\\
9.965	0.114083575885326\\
9.99	0.117033950638714\\
10	0.117910699281349\\
10.025	0.119341971521885\\
10.03	0.120093543690266\\
10.0455	0.122044957746936\\
10.15	0.130254650372299\\
10.2545	0.134399318987212\\
10.3591	0.137303482742734\\
10.4636	0.139703416216643\\
10.5682	0.141772459447027\\
10.6727	0.143570629387268\\
10.7773	0.144319152623951\\
10.8818	0.144783623795885\\
10.9864	0.14394923008241\\
11.0909	0.143212857989041\\
11.1955	0.142631573361211\\
11.3	0.14210687722437\\
11.4	0.141734004244341\\
11.5	0.141446155589162\\
11.6	0.141214388161159\\
11.7	0.141091881792966\\
11.8	0.141066544281402\\
11.9	0.14061645606149\\
12	0.139553137888542\\
12.1	0.138598852327453\\
12.2	0.137752905557644\\
12.3	0.137010385654885\\
12.4	0.136346243089098\\
12.5	0.135821654977712\\
12.6	0.135342374867745\\
12.7	0.135002903567827\\
12.8	0.134690880040991\\
12.9	0.134507379829865\\
13	0.134108400619736\\
13.1053	0.132985100326945\\
13.2105	0.132018828172555\\
13.3158	0.131124327786872\\
13.4211	0.13034379056009\\
13.5263	0.129649265157578\\
13.6316	0.129075414142791\\
13.7368421052632	0.128562402542185\\
13.8421052631579	0.128186327642971\\
13.9473684210526	0.127878096805564\\
14.0526315789474	0.127657147932328\\
14.1578947368421	0.127290034654932\\
14.2631578947368	0.125948907651761\\
14.3684210526316	0.124627405343773\\
14.4736842105263	0.123393716777005\\
14.5789473684211	0.122246323989933\\
14.6842105263158	0.121183678245351\\
14.7894736842105	0.12020420505264\\
14.8947368421053	0.11930630880346\\
15	0.118488377046045\\
15.5	0.114981182846302\\
16	0.110023920587585\\
16.5	0.106786049111517\\
17	0.103129051582506\\
17.5	0.0996707693924427\\
18	0.0976999257519676\\
18.1	0.0970576152625313\\
18.12	0.0968982960691417\\
18.13	0.0968195322645383\\
18.135	0.0967803738698635\\
18.136	0.0967725600477052\\
};

\end{axis}

\end{tikzpicture}%
\caption{Plot of $P_e$ vs. $\sfA$.}
\label{fig:P_e_Plot}
\end{subfigure}
~
\vspace{0.5cm}

\begin{subfigure}[t]{.4\linewidth}
\input{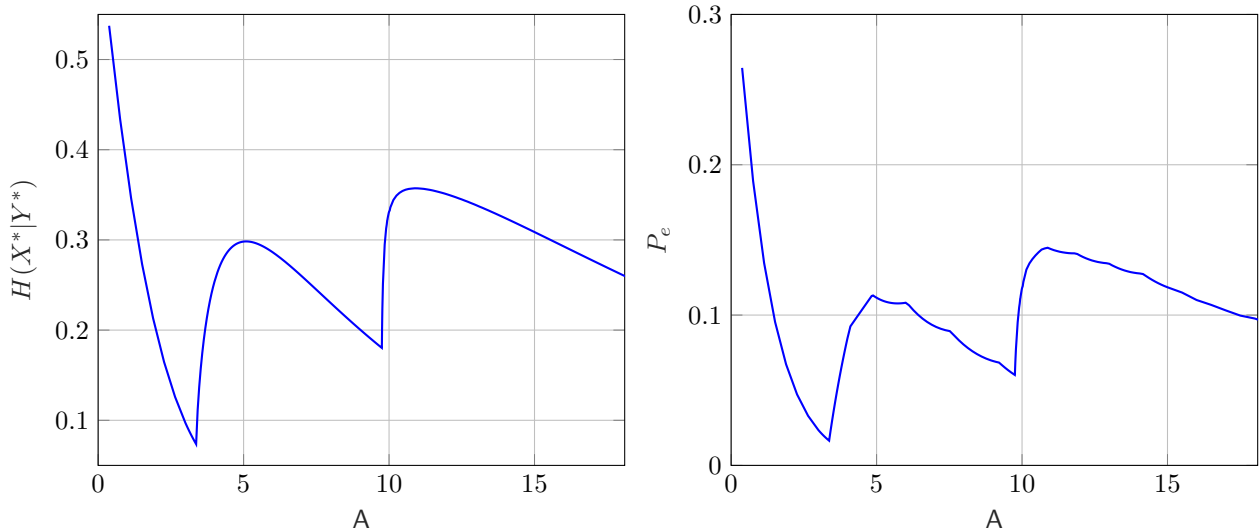}
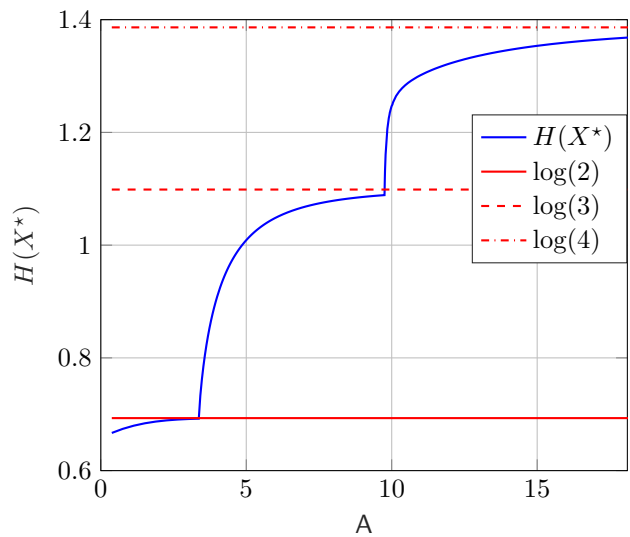
\caption{Entropy vs. $\sfA$. }
\label{fig:Entropy}
\end{subfigure}
\vspace{0.3cm}

\caption{Entropy (i.e., $H(X^\star)$), equivocation  (i.e, $H(X^\star|Y^\star)$), and the probability of error detection under MAP rule of the optimal input distribution.}
\label{fig:Plot_of_Entropy_Equivacation_and_Pe}
\end{figure}

\section{Proof of the Main Result}
\label{sec:proofs} 

The starting point for most of our proofs are the following KKT conditions  shown in  \cite{gallager1968information}; see also  the derivation in \cite{shamai1990capacity} done in the context of a Poisson noise channel. 
\begin{lemma} \label{lem:KKTconditions} The capacity-achieving distribution $P_{X^\star}$ and induced capacity-achieving output  $P_Y^\star$ distribution  satisfy the following:
\begin{subequations}
\begin{align}
&i(x; P_{X^\star})\le C(\sfA) , \quad  x \in [0, \mathsf{A}]\\
& i(x; P_{X^\star})    =  C(\sfA) , \quad  x \in \supp(P_{X^\star}),  \label{eq:EqualityCondition} 
\end{align} 
\label{eq:OptimalityMcElieceEquations}
\end{subequations}
where
\begin{equation}
	i(x; P_{X^\star})=  \sfD(P_{Y|X}(\cdot |x) \|  P_{Y^\star}  ).
	%\sum_{k=0}^\infty  P_{Y|X}(k|x) \log  \left( \frac{P_{Y|X}(k|x)}{P_{Y^\star}(k)} \right)  =\bbE \left[ \log  \left( \frac{P_{Y|X}( Y|X)}{P_{Y^\star}(Y)} \right)   |X =x  \right ]
	\end{equation}
		\end{lemma} 
	
Lemma~\ref{lem:KKTconditions} asserts that the points of the support  of $P_{X^\star}$  are zeros of the function  $ x\mapsto i(x; P_{X^\star}) -C(\sfA)$. Therefore, we have the following inclusion: 
\begin{equation}
 \supp(P_{X^\star}) 
 \subseteq  \left\{x \in [0,\sfA] :    i(x;P_{X^\star}) -C(\sfA)=0 \right \} .   \label{eq:InclusiongInequality}
\end{equation} 
This important observation will be a key step in the proof of the upper bound on the cardinality of  $ \supp(P_{X^\star}) $. 

\subsection{Proof of the New Capacity Expression in \eqref{eq:New_Cap_Expression} } 
As was shown in \cite{cao2014capacity}, $\{0\} \in \supp(P_{X^\star})$. Therefore, by using \eqref{eq:EqualityCondition}, we have that
\begin{align}
  C(\sfA) &=\sfD(P_{Y|X}(\cdot |0) \|  P_{Y^\star}  )\\
  &= \log \frac{1}{ P_{Y^\star}(0)}.
\end{align}
This concludes the proof.  Despite the simplicity of the above proof, to the best of our knowledge, the above expression has not been observed in the past.

 \subsection{Proof of the Upper Bound in \eqref{Eq:Bound_On_Probability_Values_poisson} and the Lower Bound in \eqref{eq:LowerBound_on_Number_of_Points} } 
 \label{sec:Proof_bounds_On_Probabilities} 
 
In this section, we develop upper bounds on the probability masses and a lower bound on the cardinality of the support.  Specifically, we will show that the probabilities are bounded by a term of the form $\rme^{-C(\sfA)}$, and the cardinality of the support points are lower bounded by a term of the form $\rme^{ C(\sfA) }$.

We first show an upper bound on the number of points. We show two methods for finding bounds on the probabilities. The first method relies on the strong data-processing inequality and the second method relies on the exact expression for the values of the probability distribution. 
An interesting feature of both methods is that they work  for all channels for which a capacity-achieving distribution is discrete.   In this section,  with some abuse of notation,  $P_{X^\star}$ and $P_{Y^\star}$  will denote the capacity-achieving  input and output distribution, respectively, not only for the Poisson noise channel but for a generic channel $P_{Y|X}$.

\begin{theorem}\label{thm:Bound_On_The_Number_Of_Points}    Fix some channel $P_{Y|X}$ and consider the  optimization problem 
\begin{equation}
C(\sfA)=\max_{ X \in \mathcal{X} } I(X;Y),
\end{equation} 
Moreover, suppose that a maximizing distribution $P_{X^\star}$ is  discrete. Then,
\begin{equation}
P_{X^\star}(x) \le \exp\left( - \frac{1}{  \eta_{\textsf{KL}} \left(\mathcal{X}; P_{Y|X} \right) }  C(\sfA) \right)
% \rme^{- \frac{1}{ \eta_{\textsf{KL}}(\sfA; P_{Y|X})}  C(\sfA)}
, \quad x\in   \supp(P_{X^\star}) \label{Eq:Bound_On_Probability_Values}
\end{equation} 
where    $0 <    \eta_{\textsf{KL}} \left(\mathcal{X}; P_{Y|X} \right) \le 1   $ is the contraction coefficient defined in \eqref{eq:Defintion_Of_Contraction_Coefficient}.
\end{theorem} 
\begin{proof}
Let $Y_{x}$ be the output of the channel $P_{Y|X}$ when the input is  $P_X=\delta_x$, where   $\delta_x$ is the Dirac delta function centered in $x$.  
Next, suppose that  $x\in   \supp(P_{X^\star})$. Then,   by using \eqref{eq:OptimalityMcElieceEquations}, we have that
\begin{align}
C(\sfA)&= \sfD(P_{Y|X}(\cdot |x) \|  P_{Y^\star}  )\\
&=  \sfD(P_{Y_{x}} \|  P_{Y^\star}  )\\
& \le  \eta_{\textsf{KL}} \left(\mathcal{X}; P_{Y|X} \right) \,  \sfD(\delta_x \|  P_{X^\star}  )  \label{eq:DataProcessingInequality}\\
&= \eta_{\textsf{KL}} \left(\mathcal{X}; P_{Y|X} \right)  \,\log \frac{1}{ P_{X^\star}(x)},  
\end{align}
where in \eqref{eq:DataProcessingInequality}  we have noted that $\delta_x  \to P_{Y|X}  \to P_{Y_x}$ and $P_{X^\star}  \to P_{Y|X}  \to P_{Y^\star}$ and  used the strong data-processing inequality  in \eqref{eq:SDP_def}. 
This concludes the proof.  
\end{proof}

The next result provides an upper bound on the contraction coefficient for  the Poisson channel.  
\begin{lemma} Let $ P_{Y|X}$ be a Poisson channel as in \eqref{eq:PoissonTransformation}. Then,  for all $\sfA \ge 0$
\begin{equation} 
\eta_{\textsf{KL}}([0,\sfA];P_{Y|X})  \le  1-\rme^{- \sfA }.  \label{eq:Bound_on_ContractionCoefficient}
\end{equation} 
\end{lemma}
\begin{proof}
As was shown in \cite[Proposition~II.4.10]{cohen1998comparisons}, the contraction coefficient is upper bounded by 
\begin{align}
\eta_{\textsf{KL}}([0,\sfA];P_{Y|X}) &\le  \sup_{x, x':\:  0 \le  x \le \sfA ,\, 0 \le  x' \le \sfA  } \mathsf{TV}(P_{Y|X=x} \|  P_{Y|X=x'})  \label{eq:Dobrusin_bound}\\
& \le  \sup_{x, x':\:  0 \le  x \le \sfA ,\, 0 \le  x' \le \sfA  }(  1-\rme^{- |x-x'| }) \label{eq:Bound_on_tv}\\
& =   1-\rme^{- \sfA }, 
\end{align} 
where in \eqref{eq:Dobrusin_bound}  $\mathsf{TV}(P_{Y|X=x} \|  P_{Y|X=x'}) $ is the total variation distance;  and \eqref{eq:Bound_on_tv} follows  from the bound in \cite[Corollary 3.1]{adell2005sharp}. 
 \end{proof} 
 
 Combining \eqref{Eq:Bound_On_Probability_Values} and \eqref{eq:Bound_on_ContractionCoefficient} leads to the bound in \eqref{Eq:Bound_On_Probability_Values_poisson}.  Next, we  show the proof of the second  bound in \eqref{Eq:Bound_On_Probability_Values_poisson_loc_dependent}. The following result holds for all channels that have  a discrete  capacity-achieving distribution.
  
 \begin{theorem}\label{thm:Exact_Expression} Fix some channel $P_{Y|X}$ and consider the  optimization problem 
\begin{equation}
C(\sfA)=\max_{ X \in \mathcal{X} } I(X;Y),
\end{equation} 
Moreover, suppose that a maximizing distribution $P_{X^\star}$ is  discrete.  Then,
for $x \in \supp(P_{X^\star})$
 \begin{equation}
 P_{X^\star}(x)= \rme^{-C(\sfA)} \rme^{\bbE \left[ \log \left(  P_{X^\star|Y^\star}(x|Y)\right) \middle | X=x\right]}. \label{eq:exact_Expression_for_probabilities}
 \end{equation}

\end{theorem}
\begin{proof}
By using \eqref{eq:OptimalityMcElieceEquations}, for  $x \in \supp(P_{X^\star})$ we have that
 \begin{align}
 C(\sfA)&=\sfD(P_{Y|X}(\cdot |x) \|  P_{Y^\star}  )\\
 &=\bbE \left[ \log \left(  \frac{P_{Y|X}(Y|x)}{ P_{Y^\star}(Y) }\right) \middle | X=x\right]\\
 &=\log\frac{1}{P_{X^\star}(x)}+\bbE \left[ \log \left(  P_{X^\star|Y^\star}(x|Y)\right) \middle | X=x\right].
 \end{align}
This concludes the proof.  
\end{proof} 
 
 \begin{rem} 
 Upon using the bound $P_{X^\star|Y^\star}(x|y)\ \le 1$ in \eqref{eq:exact_Expression_for_probabilities}, we arrive at 
  \begin{equation}
 P_{X^\star}(x)\le \rme^{-C(\sfA)}, \quad x \in \supp(P_{X^\star}).  \label{eq:Loose_bound_on_values}
 \end{equation}
This is  clearly a weaker version of the bound in  \eqref{Eq:Bound_On_Probability_Values}. However, as shown next, we can improve this by using a better bound on the term 
$\bbE \left[ \log \left(  P_{X^\star|Y^\star}(x|Y)\right) \middle | X=x\right]$.
 \end{rem} 
 
 \begin{lemma}
 Let $ P_{Y|X}$ be a Poisson channel as in \eqref{eq:PoissonTransformation}. Then,  for all $\sfA \ge 0$
and $x \in \supp(P_{X^\star}) \setminus \{0\}$
	\begin{equation}\label{eq:upperPx}
	P_{X^\star}(x) \le \rme^{-C(\sfA)-x \frac{\rme^{-x}}{1-\rme^{-x}}}.
	\end{equation}
	In addition, the bound in \eqref{eq:upperPx} becomes equality if $|\supp(P_{X^\star})|=2$. 
 \end{lemma} 
 \begin{proof}
	 An upper bound to the expectation in \eqref{eq:exact_Expression_for_probabilities} is as follows: for $x \in \supp(P_{X^\star})$
	 \begin{align}
	 \bbE \left[ \log \left(  P_{X^\star|Y^\star}(x|Y)\right) \middle | X=x\right]&= \sum_{k=0}^{\infty} P_{Y|X}(k|x) \log\frac{P_{Y|X}(k|x) P_{X^\star}(x)}{P_{Y^\star}(k)} \\
	 &=P_{Y|X}(0|x) \log\frac{P_{Y|X}(0|x) P_{X^\star}(x)}{P_{Y^\star}(0)}+ \sum_{k=1}^{\infty} P_{Y|X}(k|x) \log P_{X^\star | Y^\star}(x|k) \\
	 &=\rme^{-x} \log\frac{\rme^{-x} P_{X^\star}(x)}{\rme^{-C(\sfA)}}+ \sum_{k=1}^{\infty} P_{Y|X}(k|x) \log P_{X^\star | Y^\star}(x|k) \label{eq:Using_expression_for_cap}\\
	 &\le \rme^{-x} \log\frac{\rme^{-x} P_{X^\star}(x)}{\rme^{-C(\sfA)}} \label{eq:Using_the_bound P_XY le 1} \\
	 &= \rme^{-x}  \left (C(\sfA)-x+\log P_{X^\star}(x) \right),  \label{Eq:Final_Express_for_Expected_log}
	 \end{align}
	 where in \eqref{eq:Using_expression_for_cap} we have used the expression for the capacity in \eqref{eq:New_Cap_Expression}; and in \eqref{eq:Using_the_bound P_XY le 1} we have used the bound $P_{X^\star | Y^\star}(x|k) \le 1$.
	 
	 Next, for  $x \in \supp(P_{X^\star}) \setminus \{0\}$,
combining \eqref{eq:exact_Expression_for_probabilities} and \eqref{Eq:Final_Express_for_Expected_log}, 
we arrive at
\begin{equation}
P_{X^\star}(x) \le  \rme^{-C(\sfA)+ \rme^{-x}  \left (C(\sfA)-x+\log P_{X^\star}(x) \right)} , \quad x \in \supp(P_{X^\star}) \setminus \{0\}. \label{eq:Intermidiate_Step_Proof_Second_bound}
\end{equation} 
Now,   solving \eqref{eq:Intermidiate_Step_Proof_Second_bound}  for  $P_{X^\star}(x)$ yields~\eqref{eq:upperPx}. Finally, if $P_{X^\star}$ is binary, then by using Bayes rule  and the fact that $P_{Y|X}(0|0)=1$, it is not difficult to check that for $k \ge 1$
\begin{equation}
P_{X^\star | Y^\star}(\sfA|k) =1,
\end{equation}
and, therefore, the bound in \eqref{eq:Using_the_bound P_XY le 1} is tight.  This concludes the proof. 
 \end{proof} 
 
This concludes the proof of the upper bounds on the values of the probabilities.  The lower bound on the number of points in \eqref{eq:Bounds_on_the_Number_Of_Poitns} is now a consequence of the upper bound on the values of $P_{X^\star} $: 
 \begin{equation}
 1= \sum_{x \in \supp(P_{X^\star})  }  P_{X^\star}(x)  \le   |  \supp(P_{X^\star}) |  \rme^{- \frac{1}{  1-\rme^{- \sfA } }  C(\sfA)} ,
 \end{equation} 
 which simplifies to 
 \begin{equation}
   \rme^{ \frac{1}{  1-\rme^{- \sfA } }  C(\sfA)}  \le  |  \supp(P_{X^\star}) |  . \label{eq:LowerBOund_via_strong_data_processing}
 \end{equation}
 
 \begin{rem}
 An alternative way of finding a lower bound on $|  \supp(P_{X^\star}) | $ is by using the following sequence of elementary inequalities:
 \begin{align}
 C(\sfA)&= I(X^\star;Y^\star)\\
 &\le H(X^\star)\\
 &\le \log \left( |  \supp(P_{X^\star}) | \right). 
 \end{align}
 However, this upper bound is weaker than that in \eqref{eq:LowerBOund_via_strong_data_processing}. 
 \end{rem} 
 
 \subsection{Proof of the  Lower Bound  on  $P_{X^\star}(\sfA)$ in \eqref{eq:lower_bound_on_P_X(A)} and the Bounds in \eqref{eq:Bound_on_other_points} } 
 In this section, we establish bounds on  $P_{X^\star}(\sfA)$ and bounds the location of the support points.
 Let 
 \begin{align}
 i'(x ; P_{X^\star})&= \frac{ \rm d}{ {\rm d} x} i(x ; P_{X^\star}), \\
  i''(x ; P_{X^\star})&= \frac{ \rm d^2}{ {\rm d} x^2} i(x ; P_{X^\star}).
 \end{align}
\noindent The starting place for both bounds is the fact that if $x^\star \in \supp(P_{X^\star})$ is a point that is neither equal to $0$ nor $\sfA$, then  by using the KKT conditions in \eqref{eq:OptimalityMcElieceEquations}, we have that 
 \begin{align}
i'(x^\star ; P_{X^\star})=0,\\
i''(x^\star ; P_{X^\star}) \le 0.
\end{align} 
Then, by letting $\overline{Y} \sim \mathcal{P}(x^\star )$, we have that 
	\begin{align}
	0&\ge i''(x^\star ; P_{X^\star})\\
	&=-\bbE \left[ \log  \left( \bbE[X^{\star} |Y^\star=\overline{Y}+1]  \right)\right]  -i'(x^\star ; P_{X^\star})+\log(x^\star )+\frac{1}{x^\star } \label{eq:Using_identity_for_second_derivative_v1}\\
	&=-\bbE \left[ \log  \left( \bbE[X^{\star} |Y^\star=\overline{Y}+1]  \right)\right]  +\log(x^\star )+\frac{1}{x^\star } ,\label{eq:Using_Derivative_is_zeros_at_x_0} 
\end{align} 
where \eqref{eq:Using_identity_for_second_derivative_v1} follows by using the identity for the second derivative of $i(x^\star ; P_{X^\star})$ in Lemma~\ref{lem:Information_Density_Derivatives} in Appendix~\ref{app:sec:Derivatives_info_density};   and \eqref{eq:Using_Derivative_is_zeros_at_x_0} follows by using that $i'(x^\star ; P_{X^\star})=0$. 

The inequality in \eqref{eq:Using_Derivative_is_zeros_at_x_0} will be  the key to both proofs.  We start by showing the bound in \eqref{eq:Bound_on_other_points}.  We also use a result of \cite{mceliece1979practical} to establish a lower bound on the location of the second largest point. 
 \begin{lemma}\label{lem:Bound_Size_of_the_Support} Suppose that $ | \supp(P_{X^\star}) | \ge 3$. Then, for $x^\star  \in \supp(P_{X^\star})\setminus \{0,\sfA\}$
\begin{equation}
 \rme^{-\sqrt{2(\log(\sfA)-1)}} \le \sfA \, \rme^{W_{-1} \left(- \frac{1}{\sfA} \right)} \le x^\star  \le \sfA \, \rme^{W_{0} \left(- \frac{1}{\sfA} \right)} \le \sfA-1.  \label{eq:Bounds_on_smalles_and_lagrest}
\end{equation}
In addition, let $x_0 =\max \{  \supp(P_{X^\star})\setminus  \{ \sfA \} \}  $  (i.e., the second largest point in the support). Then, for  $ | \supp(P_{X^\star}) | \ge 4$,
\begin{equation}
x_0 \ge 1.  \label{eq:LowerBound_Other_Points}
\end{equation}
\end{lemma} 

\begin{proof}
By using \eqref{eq:Using_Derivative_is_zeros_at_x_0}, we have that
\begin{align}
0 
&\ge \log  \left(\frac{1 }{\sfA}  \right) +\log(x^\star )+\frac{1}{x^\star } \label{eq:using_bound_on_CE}\\
&=\log  \left(\frac{x^\star  }{\sfA}  \right) +\frac{1}{x^\star }, \label{eq:logequationtosolve}
\end{align} 
where the inequality follows by using the bound $ \bbE[X^\star|Y^\star] \le \sfA$. We now show that \eqref{eq:logequationtosolve} implies \eqref{eq:Bounds_on_smalles_and_lagrest}.

 The function $f(x)=\log  \left(\frac{x }{\sfA}  \right) +\frac{1}{x}$ is decreasing for $x<1$ and increasing for $x>1$, and has two zeros for $\sfA>\rme$. 
Note that under the assumption $ | \supp(P_{X^\star}) | \ge 3$, we have that $\sfA>\bar{\sfA}\approx 3.4$.  Therefore, we operate in the regime where we have two zeros. 

The exact solution to \eqref{eq:logequationtosolve} is given in terms of branches of the Lambert W-function:
\begin{equation}
\sfA \rme^{W_{-1} \left(- \frac{1}{\sfA} \right)} \le x^\star \le  \sfA \rme^{W_{0} \left(- \frac{1}{\sfA} \right)}.
\end{equation}
To make the above bounds more useful, we further loosen them to  involve  simpler functions.

To find the upper bound on the largest zero,  we use the bound $\log( t) \ge 1 -\frac{1}{t}$, which leads to 
\begin{align}
&0 \ge \log  \left(\frac{x^\star }{\sfA}  \right) +\frac{1}{x^\star } \ge 1- \frac{\sfA }{x^\star } +\frac{1}{x^\star } \notag \\
&\implies  x^\star  \le \sfA-1. 
\end{align}

We now find a lower bound on the smallest zero.  By substituting $x=\rme^t$, we can equivalently study the smallest zero of $f(t)=t-\log(\sfA)+\rme^{-t}$, which is the solution $t_{\min}<0$. Since $f(t)$ is decreasing for $t<0$, a lower bound to $f(t)$ will provide a lower bound to $t_{\min}$. By using $\rme^{-t}\ge 1-t+\frac{t^2}{2}$, for $t<0$, we get to study the smallest zero of
	\begin{align}
	\underline{f}(t)&=t-\log(\sfA)+1-t+\frac{t^2}{2} \\
	&=-\log(\sfA)+1+\frac{t^2}{2},
	\end{align}
	which yields
	\begin{equation}
	t_{\min}\ge -\sqrt{2(\log(\sfA)-1)},
	\end{equation}
	and finally
	\begin{equation}
	x^\star \ge \rme^{-\sqrt{2(\log(\sfA)-1)}}, \qquad \sfA>\rme.
	\end{equation}
	
	Finally, under the additional assumption that $ | \supp(P_{X^\star}) | \ge 4$ and using the fact that there is at most one point on the open interval $(0,1)$ \cite{mceliece1979practical}, it follows that the second largest  point must satisfy $x_0 \ge 1$.  This concludes the proof. 
\end{proof}

\begin{rem}  We note that from the statement of Lemma~\ref{lem:Bound_Size_of_the_Support},  the lower bound in \eqref{eq:Bounds_on_smalles_and_lagrest} is only useful to bound the location of the second smallest point since by \eqref{eq:LowerBound_Other_Points} the other larger points are lower-bounded by one.
\end{rem}

We now show the lower bound on $P_{X^\star}(\sfA)$. For ease of presentation and to emphasize key steps, the proof is split among three lemmas.

\begin{lemma}\label{lem:first_lower_bound_on_p_a}  Suppose that $ | \supp(P_{X^\star}) | \ge 3$. Let $x_0 =\max \{  \supp(P_{X^\star})\setminus  \{ \sfA \} \}  $ and $\overline{Y} \sim \mathcal{P}(x_0)$. Then, for every  $c\ge0$ such that
\begin{equation}
\bbP \left[\overline{Y}>c \right] \le \frac{1}{\sfA}, 
\end{equation}
we have that
	\begin{equation}
	P_{X^\star}(\sfA)  \ge   \frac{-\frac{1}{\sfA} \log\bigl(\frac{\sfA}{x_0}\bigr)+\frac{1}{x_0}}{\bigl(\frac{\sfA}{x_0}-1\bigr)   \bbE \left[\frac{P_{Y|X}(\overline{Y}+1|\sfA)}{P_{Y^\star}(\overline{Y}+1)} \mathds{1}_{\overline{Y}<c}\right]}\ge 0. 
	\end{equation}
\end{lemma}
\begin{proof}
	Let 	$p_\sfA =P_{X^\star}(\sfA)$. Then, starting with \eqref{eq:Using_Derivative_is_zeros_at_x_0}, we have that 
	\begin{align}
	0&\ge 
	-\bbE \left[ \log  \left( \bbE\left[\frac{X^{\star}}{x_0} \middle |Y^\star=\overline{Y}+1\right]  \right)\right] +\frac{1}{x_0} \\
	&= 
	-\bbE \left[( \mathds{1}_{\overline{Y}\ge c}+\mathds{1}_{\overline{Y}< c}) \log  \left( \bbE\left[\frac{X^{\star}}{x_0} \middle |Y^\star=\overline{Y}+1\right]  \right)\right] +\frac{1}{x_0} \\
	&\ge  
	-\bbP\left[\overline{Y}<c\right] \bbE \left[ \log  \left( \bbE\left[\frac{X^{\star}}{x_0}  \middle |Y^\star=\overline{Y}+1\right]  \right) 
	\middle | \overline{Y}< c\right] -\bbP \left[\overline{Y}\ge c \right]\log\left(\frac{\sfA}{x_0}\right)+\frac{1}{x_0} \label{eq:Bounding E[X|Y]<A} \\
	&= -\bbP \left[ \overline{Y}<c \right]\bbE \left[ \log  \left( \left(\frac{\sfA}{x_0}-1\right) \, p_\sfA \, \frac{P_{Y|X}(\overline{Y}+1|\sfA)}{P_{Y^\star}(\overline{Y}+1)}+1  \right) \middle |\overline{Y}<c\right] -\bbP\left[\overline{Y}\ge c \right]\log\left(\frac{\sfA}{x_0}\right)+\frac{1}{x_0}  \label{eq:bouning_CE_with-x_0} \\
	&\ge -\bbP \left[\overline{Y}<c \right] \log  \left( \left(\frac{\sfA}{x_0}-1\right)  \,  p_\sfA \,  \bbE \left[\frac{P_{Y|X}(\overline{Y}+1|\sfA)}{P_{Y^\star}(\overline{Y}+1)} \middle |\overline{Y}<c\right]+1  \right) -\bbP \left[\overline{Y}\ge c \right] \log\left(\frac{\sfA}{x_0}\right)+\frac{1}{x_0}, \label{eq:Using_jensens_inequality_proof_pa_lowerbound}
	\end{align}
	where in \eqref{eq:Bounding E[X|Y]<A} we have used the bound $ \bbE\left[\frac{X^{\star}}{x_0} \middle |Y^\star\right]  \le  \frac{\sfA}{x_0}$;  and \eqref{eq:bouning_CE_with-x_0}  follows by first using that
	\begin{equation}
	  \bbE\left[\frac{X^{\star}}{x_0} \middle |Y^\star\right] =\bbE\left[ \frac{X^{\star}}{x_0}  \mathds{1}_{X^{\star} \le x_0} \middle |Y^\star \right] +\bbE\left[ \frac{X^{\star}}{x_0}  \mathds{1}_{X^{\star} =\sfA } \middle|Y^\star \right] =\bbE\left[ \frac{X^{\star}}{x_0}  \mathds{1}_{X^{\star} \le x_0} \middle |Y^\star \right]+ \frac{\sfA}{x_0} \, p_\sfA \, \frac{P_{Y|X}(Y^\star|\sfA)}{P_{Y^\star}(Y^\star)},
	\end{equation} 
	and then bounding the term $\bbE\left[ \frac{X^{\star}}{x_0}  \mathds{1}_{X^{\star} \le x_0}   \middle |Y^\star \right]$ as follows:
	\begin{align}
	\bbE\left[ \frac{X^{\star}}{x_0}  \mathds{1}_{X^{\star} \le x_0} \middle |Y^\star \right]&\le  \bbE \left[ \mathds{1}_{X^{\star} \le x_0} \middle |Y^\star \right]\\
	&= \mathbb{P}\left[X^{\star} \le x_0 \middle |Y^\star \right]\\
	&=  1-  P_{X^\star|Y^\star}(\sfA|Y^\star)\\
	&=1-p_\sfA \, \frac{P_{Y|X}(Y^\star|\sfA)}{P_{Y^\star}(Y^\star)};
	\end{align}
	 and \eqref{eq:Using_jensens_inequality_proof_pa_lowerbound} follows by using Jensen's inequality. 
	
	Solving  \eqref{eq:Using_jensens_inequality_proof_pa_lowerbound} for $p_\sfA$, we obtain 
	\begin{equation}
	p_\sfA \ge \frac{\exp\left(\left(-\bbP \left[\overline{Y}\ge c \right]\log\bigl(\frac{\sfA}{x_0}\bigr)+\frac{1}{x_0}\right) \, \frac{1}{\bbP \left[\overline{Y}<c\right]}\right)-1}{\bigl(\frac{\sfA}{x_0}-1\bigr)   \bbE \left[\frac{P_{Y|X}(\overline{Y}+1|\sfA)}{P_{Y^\star}(\overline{Y}+1)} \middle |\overline{Y}<c\right]}. \label{eq:Solving_p_A_firstS-tep}
	\end{equation}
	Now, the assumption $\bbP \left[\overline{Y}>c\right] \le \frac{1}{\sfA}$ guarantees that the argument of the exponential  in \eqref{eq:Solving_p_A_firstS-tep} is positive, which leads to 
	\begin{align}
p_\sfA	&\ge\bbP \left[\overline{Y}<c \right]\frac{\exp\left(\left(-\bbP \left[\overline{Y}\ge c \right]\log\bigl(\frac{\sfA}{x_0}\bigr)+\frac{1}{x_0}\right)\frac{1}{\bbP \left[\overline{Y}<c\right]}\right)-1}{\bigl(\frac{\sfA}{x_0}-1\bigr)   \bbE \left[\frac{P_{Y|X}(\overline{Y}+1|\sfA)}{P_{Y^\star}(\overline{Y}+1)}\mathds{1}_{\overline{Y}<c}\right]}  \\
	&\ge \frac{-\bbP \left[\overline{Y}\ge c\right]\log\bigl(\frac{\sfA}{x_0}\bigr)+\frac{1}{x_0}}{\bigl(\frac{\sfA}{x_0}-1\bigr)   \bbE \left[\frac{P_{Y|X}(\overline{Y}+1|\sfA)}{P_{Y^\star}(\overline{Y}+1)}\mathds{1}_{\overline{Y}<c}\right]} \label{eq:exp>x+1 bound}\\
	&\ge \frac{-\frac{1}{\sfA}\log\bigl(\frac{\sfA}{x_0}\bigr)+\frac{1}{x_0}}{\bigl(\frac{\sfA}{x_0}-1\bigr)   \bbE \left[\frac{P_{Y|X}(\overline{Y}+1|\sfA)}{P_{Y^\star}(\overline{Y}+1)}\mathds{1}_{\overline{Y}<c}\right]}, \label{eq:Bounding_by_1/A}
	\end{align} 
	where in \eqref{eq:exp>x+1 bound} we have used the inequality  $\exp(x)\ge x+1$; and in \eqref{eq:Bounding_by_1/A} we have used that  $\bbP \left[ \overline{Y}\ge c \right] \le \frac{1}{\sfA}$.  This concludes the proof. 
\end{proof}

To find an explicit lower bound on $P_{X^*}(\sfA)$, we need to find an upper bound on $ \bbE \left[\frac{P_{Y|X}(\overline{Y}+1|\sfA)}{P_{Y^\star}(\overline{Y}+1)}\mathds{1}_{\overline{Y}<c}\right]$.  The next lemma provides such an upper bound. 

\begin{lemma} \label{lem:bound_on_the_ratio} 
Let $x_0 =\max \{  \supp(P_{X^\star})\setminus  \{ \sfA \} \}$, and suppose that the following conditions hold:
\begin{align}
|\supp(P_{X^\star})| & \ge 3,\\
1-3\:\rme^{- \frac{C(\sfA)}{ 1-\rme^{-\sfA}}  } &> 0,\\
c &<  x_0 \sfA+1 .
\end{align} 
Then,
\begin{align}
		\bbE \left[ \frac{ P_{Y|X}(\overline{Y}+1|\sfA)}{P_{Y^\star}(\overline{Y}+1)} \mathds{1}_{\bar{Y}< c} \right]  \le  \frac{1}{1-3\:\rme^{- \frac{C(\sfA)}{ 1-\rme^{-\sfA}}  }}\cdot \left\{
		\begin{array}{lc}
		x_0  \sfA^{x_0-1}, & c \le x_0 \\
		    x_0 \sfA^{x_0-1}+\sfA \rme^{-\sfA+1-x_0+c} \,\frac{1 }{\sqrt{2\pi c}(x_0 \sfA-c+1)} \left(\frac{x_0 \sfA}{c}\right)^c , & c>x_0,
		\end{array}
		\right. \label{eq:Bound_on_Renyi_type_term}
	\end{align}
	where $\overline{Y}\sim{\cal P}(x_0)$.
\end{lemma} 
\begin{proof} 
See Appendix~\ref{app:lem:bound_on_the_ratio}.
\end{proof}

The final result that we require is the following bound on the tail of the Poisson random variable.

\begin{lemma}\label{lem:bound_poisson_tail} For  $\sfA\ge \rme$ and $x_0 \ge 1$ 
\begin{equation}
\bbP \left[ \overline{Y} \ge \rme \: x_0 \log(\sfA) \right]  
\le \left( \frac{1 }{ \log(\sfA)} \right)^{ \rme  \log(\sfA)}  \rme ^{-1}\le \frac{1}{\sfA},
\end{equation}
 where   $\overline{Y}\sim{\cal P}(x_0)$.
\end{lemma} 
\begin{proof}
See Appendix~\ref{app:lem:bound_poisson_tail}. 
\end{proof}

With Lemma~\ref{lem:first_lower_bound_on_p_a}, Lemma~\ref{lem:bound_on_the_ratio}, and Lemma~\ref{lem:bound_poisson_tail} at our disposal, we are now ready  to provide an explicit bound on $P_{X^*}(\sfA)$.    Next, we make a choice of $c$ and verify that it satisfies the conditions of the aforementioned lemmas.  

To that end, choose $c=x_0 \rme \log(\sfA)$.  Moreover,  assume that $|\supp(P_{X^*})| \ge 4$, which by Lemma~\ref{lem:Bound_Size_of_the_Support} guarantees that $x_0 \ge 1$.    This choice of $c$ satisfies the condition of Lemma~\ref{lem:first_lower_bound_on_p_a}  since  $|\supp(P_{X^*})| \ge 4$ implies that $\sfA >\bar{\sfA}\approx 3.4$. Furthermore, we assume that $1-3\:\rme^{- \frac{C(\sfA)}{ 1-\rme^{-\sfA}}  } > 0 $, which is needed in Lemma~\ref{lem:bound_on_the_ratio}. 
Now using the bound in \eqref{eq:LowerBound_on_Number_of_Points}, both   $1-3\:\rme^{- \frac{C(\sfA)}{ 1-\rme^{-\sfA}}  } > 0 $ and  $|\supp(P_{X^*})| \ge 4$ hold provided that  $\rme^{ \frac{C(\sfA)}{ 1-\rme^{-\sfA}}  } \ge 4$. 
  To summarize, in order to proceed with the proof,  we require that 
\begin{align}
\rme^{ \frac{C(\sfA)}{ 1-\rme^{-\sfA}}  } &\ge 4. 
\end{align}

Having verified the conditions of Lemma~\ref{lem:first_lower_bound_on_p_a} and Lemma~\ref{lem:bound_on_the_ratio}, the bound on $P_{X^*}(\sfA)$ now proceeds as follows: 
	\begin{align}
	P_{X^*}(\sfA)
	&\ge \frac{- \frac{1}{\sfA} \log\bigl(\frac{\sfA}{x_0}\bigr)+\frac{1}{x_0}}{\bigl(\frac{\sfA}{x_0}-1\bigr)   \bbE \left[\frac{P_{Y|X}(\overline{Y}+1|\sfA)}{P_{Y^\star}(\overline{Y}+1)}\mathds{1}_{\overline{Y}<c}\right]} \label{eq:First_step_bounding_explicit_p_a}\\
&\ge	\frac{\sfA^{-2}}{ \bbE \left[\frac{P_{Y|X}(\overline{Y}+1|\sfA)}{P_{Y^\star}(\overline{Y}+1)}\mathds{1}_{\overline{Y}<c}\right]} \label{eq:using_monotonicity_ -(1/A)ln(A/x)+1/x}\\
	&\ge  \frac{\sfA^{-2} \Bigl (1-3\:\rme^{- \frac{C(\sfA)}{ 1-\rme^{-\sfA}}  } \Bigr)}{   x_0 \sfA^{x_0-1}+\sfA \rme^{-\sfA+1-x_0+c} \, \frac{1 }{\sqrt{2\pi c}(x_0 \sfA-c+1)} \left(\frac{x_0 \sfA}{c}\right)^c } \label{eq:bounding_the_denominator}\\
	&\ge  \frac{\sfA^{-2}  \Bigl(1-3\:\rme^{- \frac{C(\sfA)}{ 1-\rme^{-\sfA}}  } \Bigr) }{x_0  \sfA^{x_0-1}+\sfA \rme^{-\sfA+1-x_0+x_0\rme  \log(\sfA)} \, \bigl(\frac{x_0 \sfA}{x_0\rme  \log(\sfA)}\bigr)^{x_0\rme  \log(\sfA)} } \label{eq:bounding_with_choice_c} \\
	&\ge  \frac{1-3\:\rme^{- \frac{C(\sfA)}{ 1-\rme^{-\sfA}}  } }{2\sfA^{2\sfA \rme  \log(\sfA)+2} \, \rme^{-2\sfA+1}  } ,\label{eq:Bounding_the_messy_algebra}
	\end{align}
where \eqref{eq:First_step_bounding_explicit_p_a} is due to the bound in Lemma~\ref{lem:first_lower_bound_on_p_a};	 \eqref{eq:using_monotonicity_ -(1/A)ln(A/x)+1/x} follows from the upper bound $\frac{\sfA}{x_0}-1 \le \sfA$ because $x_0\ge 1$, and the fact  $x \mapsto -\frac{1}{\sfA} \log\left(\frac{\sfA}{x}\right)+\frac{1}{x}$ is minimized at $x=\sfA$; \eqref{eq:bounding_the_denominator} follows by using the bound in \eqref{eq:Bound_on_Renyi_type_term}; \eqref{eq:bounding_with_choice_c} follows from the choice of $c=x_0 \rme \log(\sfA)$  and the bound $\frac{1 }{\sqrt{2\pi c}(x_0 \sfA-c+1)} \le 1$; and \eqref{eq:Bounding_the_messy_algebra} follows by using $1 \le x_0 \le \sfA$, which implies that 
	\begin{equation}
	x_0  \sfA^{x_0-1}+\sfA \rme^{-\sfA+1-x_0+x_0\rme  \log(\sfA)} \,\left(\frac{x_0 \sfA}{x_0\rme  \log(\sfA)}\right)^{x_0\rme  \log(\sfA)} \le 2\sfA^{2\sfA \rme  \log(\sfA)} \, \rme^{-2\sfA+1} . 
	\end{equation}  
This concludes the proof of the lower bound on $P_{X^*}(\sfA)$. 

\begin{rem} To arrive at the bound in \eqref{eq:Bounding_the_messy_algebra}, we have used $x_0 \ge 1$, which requires an assumption that $\rme^{ \frac{C(\sfA)}{ 1-\rme^{-\sfA}}  } \ge 4$.  Alternatively, we could have used the bound  $x_0 \ge  \rme^{-\sqrt{2(\log(\sfA)-1)}} $ in \eqref{eq:Bounds_on_smalles_and_lagrest}, which only requires an assumption that  $\rme^{ \frac{C(\sfA)}{ 1-\rme^{-\sfA}}  } > 3$, and would result in 
\begin{equation}
P_{X^*}(\sfA) \ge \frac{1-3\:\rme^{- \frac{C(\sfA)}{ 1-\rme^{-\sfA}}  } }{2\sfA^{2\sfA \rme  \log(\sfA) \exp \left(-\sqrt{2(\log(\sfA)-1)} \right)+2} \,  \exp \left( -2\sfA+1 \right)  } .
\end{equation}
However, asymptotically this bound is weaker than that  in  \eqref{eq:Bounding_the_messy_algebra}.
\end{rem}

\subsection{Proof of the Upper Bound in \eqref{eq:Bounds_on_the_Number_Of_Poitns}}
 In this section, we establish an upper bound on the number of the support points.  Specifically, we will establish a bound of the order $\sfA  \log^2(\sfA)$. To show the bound on the number of points, we need to first present a number of ancillary results. The first result that we will need allows us to bound the number of zeros of $f$ with the number of zeros of $f'$. 
\begin{lemma}\label{lem:extreme points and oscillations}
Suppose that $f$ is continuous  on $[-R,R]$ and differentiable on $(-R,R)$. If $\rmN([-R,R],f) <\infty$, then 
\begin{equation}
  \rmN([-R,R],f) \le   \rmN([-R,R],f ') + 1 \text{,}
\end{equation}
where $f'$ denotes the derivative of $f$.  
\end{lemma}
\begin{proof}
Let $x_1 < \dots <x_n$ denote the zeros of $f$. By Rolle’s Theorem, each of the intervals $(x_i,x_{i+1})$ for $i=1, \ldots ,n -1$ contains at least one extreme point.
\end{proof}

The final ancillary result that we need is the following lemma.  

\begin{lemma} \label{lem:BoundOnNumberOfOsillations}
Let $f:\bbN_0 \to \mathbb{R}$.  Suppose that there exists  $k^\star$ such that  for all $k \ge  k^\star$  either
\begin{equation}
f(k)  \le   0  \text{ or }    f(k)  \ge   0.  
\end{equation} 
Then,
\begin{equation}
 \scrS \left(  f      \right)  \le k^\star. 
\end{equation} 
\end{lemma} 
\begin{proof}
To prove the bound, note that for all  $ k \ge k^\star$ the function $f(k)$ no longer  can change sign. Therefore, there can be at most $k^\star$ sign changes. 
\end{proof}

The next result provides an upper bound on the number of support points in terms of the number of sign changes of the function related to $P_{Y^\star}$.

\begin{lemma}  For every $\sfA>0$
\begin{equation}
|\supp(P_{X^\star})|   \le   \scrS \left(   \psi(  \boldsymbol{\cdot} ; P_{X^\star })   \right)   +2, \label{eq:FirstMainBound}
\end{equation}
where  
\begin{equation}
  \psi \left( k ; P_{X^\star } \right) = k \log   \left(\frac{1}{\bbE[X^\star |Y^\star=k-1]}  \right)+  \log  \left( k!P_{Y^\star}(k)  \right) + C(\sfA) + k,  \quad k\in \mathbb{N}_0, 
\end{equation}
where we define $k \log   \left(\frac{1}{\bbE[X^\star |Y^\star=k-1]}  \right) |_{k=0}=0$. 
\end{lemma} 

\begin{proof}
First, using Lemma~\ref{lem:Information_Density_Derivatives} in Appendix~\ref{app:sec:Derivatives_info_density}, we have that 
\begin{equation}
	i(x; P_{X^\star})-C(\sfA)= G(x)  -x-C(\sfA) +x \log(x),    \label{eq:Definition_of_functiong_G}
\end{equation}
where
\begin{align}
  G(x) &= \sum_{k=0}^\infty  P_{Y|X}(k|x)  \log  \left(\frac{1 }{k!P_{Y^\star}(k)}  \right)\\
   &= \sum_{k=0}^\infty  P_{Y|X}(k|x)  g(k) , \label{eq:expression_for_G}
\end{align}
and where we let $g(k) =\log  \left(\frac{1 }{k!P_{Y^\star}(k)}  \right)$.

Now, the upper bound can be established as follows: 
\begin{align}
| \supp(P_{X^\star})| & \le   \rmN \left([0,\sfA], 	i(x; P_{X^\star})-C(\sfA) \right)  \label{eq:Containment}\\
& =   \rmN \left((0,\sfA], 	i(x; P_{X^\star})-C(\sfA) \right)  +1 \label{eq:Using_that_0_is_apoint_of_support}\\
& =   \rmN \left((0,\sfA], 	 G(x)  -x-C(\sfA) +x \log(x) \right)  +1 \label{eq:Using_Defintion_G}\\
& =   \rmN \left((0,\sfA], 	 \frac{G(x)}{x}  -1- \frac{C(\sfA)}{x} + \log(x) \right)  +1 \label{Eq:Dividing_by_x}\\
& \le   \rmN \left((0,\sfA], 	 \frac{G'(x)}{x}- \frac{G(x)}{x^2}  + \frac{C(\sfA)}{x^2} +  \frac{1}{x} \right)  +2 \label{eq:apply_Rolls_Thm}\\
& =   \rmN \left((0,\sfA], 	 x  G'(x)-  G(x) + C(\sfA) +  x \right)  +2  \label{Eq:Multiplying _by_x^2} \\
& =   \rmN \left((0,\sfA], 	  \bbE \left[  Y ( g(Y)-g(Y-1)) -  g(Y) +C(\sfA) +Y  \middle| X=x \right ]  \right)  +2 \label{eq:UsingAuxiliaryLemma}\\
& =   \rmN \left((0,\sfA], 	  \bbE \left[  \psi(Y;P_{X^\star}) \middle | X=x\right ]  \right)  +2  \label{eq:define_psi}\\
& \le \scrS \left(     \psi(\boldsymbol{\cdot};P_{X^\star}) \right) +2, \label{eq:Applying_Oscilation_Theorem}
\end{align}
where  \eqref{eq:Containment} follows from the containment in \eqref{eq:InclusiongInequality};   \eqref{eq:Using_that_0_is_apoint_of_support} follows by using the fact that $0\in \supp(P_{X^\star})$; \eqref{eq:Using_Defintion_G} follows by using the definition of $G(x)$ in \eqref{eq:Definition_of_functiong_G}; \eqref{Eq:Dividing_by_x} follows by using the fact that  $ G(x)  -x-C(\sfA) +x \log(x) $ and $ \frac{G(x)}{x}  -1- \frac{C(\sfA)}{x} + \log(x)$ have the same zeros on $(0,\sfA]$; \eqref{eq:apply_Rolls_Thm} follows by using the bound in Lemma~\ref{lem:extreme points and oscillations}; \eqref{Eq:Multiplying _by_x^2} follows by using the fact that   $\frac{G'(x)}{x}- \frac{G(x)}{x^2}  + \frac{C(\sfA)}{x^2} +  \frac{1}{x}$ and  $ x  G'(x)-  G(x) + C(\sfA) +  x $ have the same number of zeros on $(0,\sfA]$;  \eqref{eq:UsingAuxiliaryLemma} follows by  using Lemma~\ref{lem:auxiliaryLemma} in Appendix~\ref{app:sec:Derivatives_info_density}  which leads to 
\begin{align}
  x  G'(x) &= x  \bbE \left[  g(Y+1)-g(Y) \middle | X=x \right ]=   \bbE \left[ Y ( g(Y)-g(Y-1)) \middle | X=x  \right],  \label{eq:Using_prodcut_differenc_identity} \\
   -G(x) + C(\sfA) +  x &=    \bbE \left[ - g(Y) + C(\sfA) +  Y \middle | X=x  \right],
  \end{align} 
    where in \eqref{eq:Using_prodcut_differenc_identity} we assume that $k g(k-1)|_{k=0}=0$;
  \eqref{eq:define_psi} follows by defining 
  \begin{align}
  \psi(k;P_{X^\star})&=  k ( g(k)-g(k-1)) -  g(k) + C(\sfA) +  k\\
  &=k \log  \left(\frac{ (k-1)! P_{Y^\star}(k-1) }{ k! P_{Y^\star}(k)}  \right)+  \log  \left( k!P_{Y^\star}(k)  \right) + C(\sfA) + k\\
  &=k \log   \left(\frac{1}{\bbE[X^\star|Y^\star=k-1]}  \right)+  \log  \left( k!P_{Y^\star}(k)  \right) + C(\sfA) + k,
    \end{align}
     and  \eqref{eq:Applying_Oscilation_Theorem} follows from the oscillation theorem.
This concludes the proof.  
\end{proof} 

\begin{rem} Note that the goal of steps  \eqref{eq:Containment} through \eqref{eq:define_psi} is to transform the function $x \mapsto i(x; P_{X^\star})-C(\sfA)$ to a function of the form $\bbE \left[  \psi(Y;P_{X^\star}) | X=x\right ]$. The latter formulation is significant as it enables the application of the oscillation theorem.  The reader might wonder why such a representation was not possible right after \eqref{eq:Containment}?   To see  why this cannot be done, note that 
\begin{align}
  i(x; P_{X^\star})-C(\sfA) &=G(x)  -x-C(\sfA) +x \log(x) \\
  &=\bbE \left[g(Y)  -Y-C(\sfA) \middle|X=x \right]  +x \log(x),
\end{align} 
where in the last step we have used the expression for $G(x)$ in \eqref{eq:expression_for_G} and the fact that $\bbE[Y|X=x]=x$.  Therefore, it remains to answer whether we can write $\bbE[h(Y)|X=x]=x \log(x) $ for some function $h$. In other words, can we show that $x\log (x) $ is a statistic of a Poisson distribution with mean $x$? To answer this question,  towards a contradiction, assume that   $x \log(x) $ is a statistic of a Poisson distribution and there exists a function $h: \bbN_0 \to \mathbb{R}$ such that
\begin{equation}
x \log(x)=\bbE\left[h(Y)|X=x \right]= \sum_{k=0}^\infty  \frac{h(k)}{k!} x^k  \rme^{-x}. 
\end{equation}
Multiplying the above by   $\rme^{x}$,  we arrive at
\begin{equation} 
 \rme^{x} x \log(x)=\sum_{k=0}^\infty  \frac{h(k)}{k!} x^k, \label{eq:Power_Series_Conctradition}
 \end{equation}
 which asserts that the function  $\rme^{x} x \log(x)$ has a power series  representation at $x=0$. This is, of course, not possible and leads to a contradiction.  Therefore, we cannot write $\bbE[h(Y)|X=x]=x \log(x) $.  The purpose of steps \eqref{eq:Containment} through \eqref{eq:define_psi} was  to  bypass this   issue and eliminate the term $x \log(x)$ while keeping the number of zeros relative unchanged.   Finally, we note that the analyticity argument leading to the contradiction in \eqref{eq:Power_Series_Conctradition} has been used before in \cite[Thm.~15]{cheraghchi2018improved} to show  discreteness of the capacity-achieving distribution for the Poisson noise channel with both  peak-power  and  average-power constraints.    
  \end{rem}

The next result provides an upper bound on the number of sign changes $  \scrS \left(   \psi(  \boldsymbol{\cdot} ; P_{X^\star })   \right)$. 
\begin{lemma} For $\sfA>0$
\begin{equation}
\scrS \left(   \psi(  \boldsymbol{\cdot} ; P_{X^\star })   \right) \le   k^\star 
\end{equation}
where
\begin{equation}
 k^\star = \lceil \sfA  - \log  \left( P_{X^\star}(\sfA)\right) - C(\sfA) \rceil \ge 0.  \label{eq:k_star_bound}
\end{equation}

\end{lemma} 
\begin{proof}
We begin by lower bounding $ \psi(  \boldsymbol{\cdot} ; P_{X^\star })$
 \begin{align}
  \psi(k; P_{X^\star})  &=k \log   \left(\frac{1}{\bbE[X^\star |Y^\star=k-1]}  \right)+  \log  \left( k!P_{Y^\star}(k)  \right) + C(\sfA) + k \\
  & \ge  k \log   \left(\frac{1}{\sfA}  \right)+  \log  \left( k!P_{Y^\star}(k)  \right) + C(\sfA) + k \label{Eq:Bounding_Conditional_EXP}\\
    & \ge  k \log   \left(\frac{1}{\sfA}  \right)+  \log  \left( P_{X^\star}(\sfA) \sfA^k \rme^{-\sfA} \right) + C(\sfA) + k\label{eq:LowerBoundon_on_P_Y}\\
    &=- \sfA  + \log  \left( P_{X^\star}(\sfA) \right) + C(\sfA) + k,
    \end{align} 
    where \eqref{Eq:Bounding_Conditional_EXP}  follows by using the bound  $\bbE[X^\star|Y^\star] \le \sfA$; and \eqref{eq:LowerBoundon_on_P_Y} follows by using the fact that there is a mass point at $\sfA$ which leads to the bound
    \begin{align}
    P_{Y^\star}(k)= \frac{1}{k!} \bbE \left[ {X^\star}^k \rme^{-X^\star} \right]  \ge   \frac{1}{k!}   P_{X^\star}(\sfA)  \sfA^k \rme^{-\sfA}. 
    \end{align} 
    
    Therefore, there exists  $k^\star = \lceil \sfA  - \log  \left( P_{X^\star}(\sfA) \right) - C(\sfA) \rceil$ such that
    \begin{equation}
      \psi(k; P_{X^\star}) \ge 0,   \,    k\ge   k^\star . 
    \end{equation}
    It remains to show that $k^\star$ is larger than zero.   This follows by using the bound in \eqref{Eq:Bound_On_Probability_Values}
    \begin{align}
    \sfA  - \log  \left( P_{X^\star}(\sfA) \right) - C(\sfA)   \ge       \sfA  - \log  \left(\rme^{-C(\sfA) }\right) - C(\sfA) =\sfA
    \end{align} 
    The proof is concluded by using Lemma~\ref{lem:BoundOnNumberOfOsillations}. 
 \end{proof}
 
 \begin{rem} Note that value of $P_{X^\star}(\sfA)$ was introduced in the step \eqref{eq:LowerBoundon_on_P_Y}.  Since a significant part of this work was dedicated to find a lower bound on $P_{X^\star}(\sfA)$, it is interesting to discuss how tight  the bound in \eqref{eq:LowerBoundon_on_P_Y} actually is.  Observe that
 \begin{align}
 \lim_{k \to \infty}  \frac{k! P_{Y^\star}(k)}{ P_{X^\star}(\sfA) \sfA^k \rme^{-\sfA}  }=   \lim_{k \to \infty} \sum_{x  \in \supp(P_{X^\star}) \setminus \{\sfA \} }  \frac{P_{X^\star}(x)}{P_{X^\star}(\sfA)} \left( \frac{x}{\sfA} \right)^k \rme^{-x+\sfA} +1=1.
 \end{align}
 Therefore, asymptotically the bound used to produced the inequality in \eqref{eq:LowerBoundon_on_P_Y} is tight. 
 \end{rem} 
 
 Now, combining the bounds in \eqref{eq:FirstMainBound}, \eqref{eq:k_star_bound} and  \eqref{eq:lower_bound_on_P_X(A)}, we arrive at
 \begin{align}
|\supp(P_{X^\star})|  &\le \lceil \sfA  - \log  \left( P_{X^\star}(\sfA)\right) - C(\sfA) \rceil +2\\
 &\le  \sfA  - \log  \left( P_{X^\star}(\sfA)\right) - C(\sfA)  +3 \\
   %&\le  \sfA  - \log  \left(  \frac { 1-  \rme^{- \frac{C(\sfA)}{ 1-\rme^{-\sfA}}  }  }{  \sfA^2 \exp \left( -\sfA+\sfA (\sfA-1)\rme^{\sqrt{2(\log(\sfA)-1)}} \right) } \right) - C(\sfA)  +1 \\
   &\le  \sfA  - \log  \left(  \frac{1-3\:\rme^{- \frac{C(\sfA)}{ 1-\rme^{-\sfA}}  } }{2\sfA^{2\sfA \rme  \log(\sfA)+2} \, \rme^{-2\sfA+1}  } \right) - C(\sfA)  +3\\
   &= 2  \rme \sfA  \log^2(\sfA)+2\log(\sfA)-\sfA- \log  \left(  \frac{1-3\:\rme^{- \frac{C(\sfA)}{ 1-\rme^{-\sfA}}  } }{2 } \right) - C(\sfA)  +4,
    \end{align}
 where the first inequality is due to the bound $\lceil x \rceil \le x+1$.
 This concludes the proof.

 \section{Conclusion} 
 \label{sec:Conclusion}
 
This work has focused on studying properties of the capacity-achieving distribution for the Poisson noise channel with an amplitude constraint.  
 It was previously known that the capacity-achieving distribution for this channel is discrete with finitely many points.  In this work, we sharpened this result in several ways.   
 
 First, by using a strong data-processing inequality,  an upper bound on the values of the mass points has been shown. This upper bound on the probability values has been shown to lead to the lower bound on the number of  support point of the optimal input distribution. Specifically,  a  lower bound  of order  $\sqrt{\sfA}$ has been established on the number of support points where $\sfA$ is the constraint on the amplitude. 
 
 Second, by using the variation-diminishing property of the Poisson kernel, the work has also established an upper bound on the number of the support points of the optimal input distribution. Specifically, an order $\sfA \log^2(\sfA)$ bound has been established. 
 
 Finally, along the way, several other results have been shown. For example, a new compact expression for the capacity has been shown. In addition, a lower bound on the probability of the largest points of the optimal input distribution has been established.  Furthermore, an estimate on the locations of the support other than $0$ and $\sfA$ has been established. 

 There are few interesting future directions and open questions:
\begin{itemize}
\item Some of the ideas in this work, in particular the oscillation theorem, were borrowed from \cite{dytso2019capacity} where a Gaussian noise channel was considered. It is interesting to note that in \cite{dytso2019capacity}, it was shown that the oscillation theorem leads to an order tight bound on the cardinality of support of the capacity-achieving distribution. Currently, we do not have such a claim for the Poisson noise channel. Therefore, an interesting future direction will be to assess whether the bound due to the oscillation theorem is also order tight for the Poisson noise channel. 
\item The location of the second-largest point plays an important role in our analysis. In particular, it will be interesting to investigate the gap between the largest and the second-largest points. Currently, this gap is lower-bounded by a constant. In other words, the second-largest point never approaches $\sfA$.  In fact, the numerical simulation suggests that this gap increases as a function of $\sfA$.  One interesting future direction is to investigate this gap and determine whether it increases with $\sfA$. 

\item This work has developed a method for finding lower bounds on the value of $P_{X^\star}(\sfA)$. This lower bound was instrumental for finding an explicit upper bound on the cardinality of support.   However, the current bound on $P_{X^\star}(\sfA)$ appears to be loose.  An important future direction is to find ways of improving this bound. 
\item The general upper bounds on the values of the probabilities can be applied to other channels, and it would be interesting to see how tight these bounds are.  Another interesting future direction is to extend the results of this paper to a more general Poisson model that includes the dark current parameter. 
 \end{itemize}

 \begin{appendices}

  \section{Sketch of the Implementation of the Gradient Ascent} 
  \label{app:Gradient_Ascent}
  
  First, to apply the projected gradient ascent \cite{shalev2014understanding},  we first parametrize the input distribution as a vector in $\mathbb{R}^{2n}$
\begin{equation}
P_X  \Longrightarrow     {\bf x}=[{\bf x}_1, {\bf x}_2] =[ \hspace{-0.15cm} \underbrace{x_1,x_2, \ldots, x_n,}_{ \text{points of the support}} \, \,\, \underbrace{p_1, p_2, \ldots, p_n}_{\text{probability values}}].
\end{equation}
From Theorem~\ref{thm:Main_Result}, for a given $\sfA$,  we know that $n$ can be at most  equal to \eqref{eq:Bounds_on_the_Number_Of_Poitns}. Therefore, we set $n$ to be equal to \eqref{eq:Bounds_on_the_Number_Of_Poitns}.  Second, we explicitly write the mutual information as a function of ${\bf x}$ and let
\begin{equation}
g( {\bf x}) =I(X;Y).
\end{equation}
In view of Theorem~\ref{thm:Main_Result}, we have that
\begin{align}
\max_{ P_X: 0 \le X \le \sfA} I(X;Y) = \max_{{\bf x}:\: {\bf x}_1\in [0, \sfA]^{n},\,  {\bf x}_2 \in \mathcal{P}  } g({\bf x}), 
\end{align} 
where $\mathcal{P}$ is the probability simplex.    

 The update equation for the projected  gradient ascent are now given by 
 \begin{equation}
 {\bf x}^{(t+1)}= \mathsf{proj} \left(  {\bf x}^{(t)}+\lambda  \nabla g( {\bf x}^{(t)}) \right), \, t\in \mathbb{N}
 \end{equation} 
 where $\lambda>0$ is some step size (to generate Fig.~\ref{fig:Plot_of_probabilitiues}, we have used $\lambda=0.01$),   $\nabla g({\bf x})$ is the gradient of $g$, and  the projection operation $\mathsf{proj} $ maps a vector ${\bf x}=[{\bf x}_1,{\bf x}_2]$ to the set 
 \begin{equation}
\bigl \{ {\bf x}: {\bf x}_1\in [0, \sfA]^{n},  {\bf x}_2 \in \mathcal{P}   \bigr \}.
 \end{equation} 
The projection of the first element ${\bf x}_1$ onto a cube  $ [0, \sfA]^{n}$ can be done by individually going through each coordinate of ${\bf x}_1$.   An efficient implementation of the projection of  ${\bf x}_2$ onto the probability simplex can be found in \cite{wang2013projection}. 
  
  Finally, the initial condition ${\bf x}^{(1)}$ for given $\sfA$, is chosen to be the distribution  that was a final output of the gradient algorithm for the previous epoch in which we considered $\sfA'<\sfA$.  For the best performance, the difference between $\sfA'$ and $\sfA$ should be made as small as possible.

  \section{Derivatives of $i(\cdot;P_{X})$}
  \label{app:sec:Derivatives_info_density} 
To find the derivatives of the  term $i(\cdot;P_{X})$, we will need the following auxiliary results, the proof of which can be found in \cite{dytso2020estimation}. 
 \begin{lemma}\label{lem:auxiliaryLemma} Let $f : \bbN_0 \to \mathbb{R}$  and assume $P_{Y|X}$ is given in \eqref{eq:PoissonTransformation}. Then,  for any $P_X$,  the following identities hold:
\begin{align}
%x\bbE[  f(Y) | X=x ]&= \bbE[  Yf(Y-1) | X=x ],\\ \label{eq:PoissonShiftEquation}
\frac{\rm d}{ {\rm d} x}   \bbE[  f(Y) | X=x ]&=  \bbE[  f(Y+1)-f(Y) | X=x ],\quad x>0,\label{eq:DerivativePropety} \\
 x \bbE[  f(Y) | X=x ]&= \bbE[ Y f(Y-1)|X=x],\quad x \ge 0
\end{align}
 where $k f(k-1)|_{k=0}=0$ and 
\begin{align}
\bbE[X|Y=y] &=\frac{(y+1) P_Y(y+1)}{P_Y(y)},\quad  y \in \bbN_0,  \label{Eq:EmpricalBayes}\\
\bbE[X^k|Y=y]&=\prod_{i=1}^{k-1}  \bbE[X|Y=y+i],\quad k \in \bbN, \, y \in \bbN_0.  \label{eq:Hiher_moments}
\end{align} 
\end{lemma}

The derivatives of $i(\cdot;P_{X})$ are given next. 

\begin{lemma}
\label{lem:Information_Density_Derivatives} We have
  \begin{equation}
	i(x; P_{X})%&=  \sum_{k=0}^\infty  P_{Y|X}(k|x) \left(  \log  \left(\frac{1 }{P_{Y^\star}(k)}  \right)  +\log \left( \frac{1}{k!} \right)   \right) + x \log(x) -x \\
	= G(x) +  x \log(x) -x,   \quad x\ge 0,  \label{eq:Decomosition_onto_G_and_xlogx}
	\end{equation}
where
	\begin{equation}
		G(x) = \sum_{k=0}^\infty  P_{Y|X}(k|x)  \log  \left(\frac{1 }{k!P_{Y}(k)}  \right) . 
	\end{equation} 
Moreover,
\begin{itemize}
\item The first and second derivative  of $G(x)$ are given by 
\begin{align}
G'(x)  &= \sum_{k=0}^\infty  P_{Y|X}(k|x)  \log  \left(\frac{ 1 }{ \bbE[X|Y=k]}  \right), \\
G''(x)  &= \sum_{k=0}^\infty  P_{Y|X}(k|x)  \log  \left(\frac{ \bbE[X|Y=k] }{ \bbE[X|Y=k+1]}  \right) \\
&= \sum_{k=0}^\infty  P_{Y|X}(k|x)  \log  \left(\frac{ \bbE^2[X|Y=k] }{ \bbE[X^2|Y=k]}  \right) \\
&=\sum_{k=0}^\infty  P_{Y|X}(k|x)  \log  \left(\frac{1 }{ \bbE[X|Y=k+1]}  \right) -G'(x)  .
%G'''(x) 
%&= \sum_{k=0}^\infty  P_{Y|X}(k|x)  \log  \left(\frac{ \bbE[X|Y=k+1] }{\bbE[X|Y=k+2] }  \right) -G''(x)\\
%&= \sum_{k=0}^\infty  P_{Y|X}(k|x)  \log  \left(\frac{ 1}{\bbE[X|Y=k+2] }  \right) -G'(x)-2G''(x)
\end{align} 
%\item $G(x)$ is strictly concave. 
%\item   $\sum_{k =2}^n G^{(k)}(x) <0$,  for every $n \ge 2$. 
\item The second derivative of $i(x; P_X)$ is given by 
\begin{align}
i''(x; P_{X})
&=\sum_{k=0}^\infty  P_{Y|X}(k|x)  \log  \left(\frac{1 }{ \bbE[X|Y=k+1]}  \right) -i'(x; P_{X})+\log(x)+\frac{1}{x}. \label{eq:infor_density_second_derivative}
\end{align}

\end{itemize}

\end{lemma}

\begin{proof}

We first show the decomposition in \eqref{eq:Decomosition_onto_G_and_xlogx} 
\begin{align}
	i(x; P_{X})&= \sum_{k=0}^\infty  P_{Y|X}(k|x)  \log  \left( \frac{P_{Y|X}(k|x) }{P_{Y}(k)}  \right)  \\
	&= \sum_{k=0}^\infty  P_{Y|X}(k|x) \left(  \log  \left(\frac{1 }{P_{Y}(k)}  \right)  +\log \left( \frac{1}{k!} \right)  + k \log(x) -x \right) \label{eq:re-writing_info_density_step1} \\
	&= \sum_{k=0}^\infty  P_{Y|X}(k|x) \left(  \log  \left(\frac{1 }{P_{Y}(k)}  \right)  +\log \left( \frac{1}{k!} \right)   \right) + x \log(x) -x \label{eq:re-writing_info_density_step2} \\
	&= G(x)  -x +x \log(x),   \label{eq:re-writingi(x,Px)}
\end{align}
where \eqref{eq:re-writing_info_density_step1} follows by using that  
\begin{align}
\log \left( P_{Y|X}(k|x)  \right)= \log \left( \frac{x^k}{k!} \rme^{-x} \right)= \log \left( \frac{1}{k!} \right)  + k \log(x) -x; 
\end{align} 
and where \eqref{eq:re-writing_info_density_step2}  follows by using that $\bbE[Y|X=x]=x$.

Next, to find the derivative of $G(x)$, we use the properties in \eqref{eq:DerivativePropety} and \eqref{Eq:EmpricalBayes} which lead to 
\begin{align}
G'(x) &= \sum_{k=0}^\infty  P_{Y|X}(k|x)  \log  \left(\frac{ P_{Y}(k) }{(k+1)P_{Y}(k+1)}  \right) \\
&= \sum_{k=0}^\infty  P_{Y|X}(k|x)  \log  \left(\frac{ 1 }{ \bbE[X|Y=k]}  \right). 
\end{align}
Similarly, by using \eqref{eq:DerivativePropety} and \eqref{eq:Hiher_moments} we arrive at
\begin{align}
G''(x)  &= \sum_{k=0}^\infty  P_{Y|X}(k|x)  \log  \left(\frac{ \bbE[X|Y=k] }{ \bbE[X|Y=k+1]}  \right) \\
&= \sum_{k=0}^\infty  P_{Y|X}(k|x)  \log  \left(\frac{ \bbE^2[X|Y=k] }{ \bbE[X^2|Y=k]}  \right) .
\end{align}

%The third derivative is given by
%\begin{align}
%G'''(x) &= \sum_{k=0}^\infty  P_{Y|X}(k|x) \log  \left(\frac{ \bbE^2[X|Y=k+1]  \bbE[X^2|Y=k] }{ \bbE[X^2|Y=k+1]  \bbE^2[X|Y=k] }  \right) \\
%&= \sum_{k=0}^\infty  P_{Y|X}(k|x)  \log  \left(\frac{ \bbE^2[X|Y=k+1] }{\bbE[X|Y=k+2] \bbE[X|Y=k]}  \right) \\
%&= \sum_{k=0}^\infty  P_{Y|X}(k|x)  \log  \left(\frac{ \bbE[X|Y=k+1] }{\bbE[X|Y=k+2] }  \right) -G''(x)
%\end{align} 
%
%Next, note that
%\begin{align}
%e^{x} G''(x)  &= \sum_{k=0}^\infty  \frac{x^k}{k!}  \log  \left(\frac{ \bbE[X|Y=k] }{ \bbE[X|Y=k+1]}  \right) \\
%&= \sum_{k=0}^\infty  \frac{x^k}{k!}  a_k
%\end{align} 
%where $a_k \le 0$.    Therefore, every derivative of $e^{x} G''(x)<0$.  Using the product rule, this implies that for every $n$
%\begin{align}
%e^x\sum_{k =2}^n G^{(k)}(x) <0 \Rightarrow  \sum_{k =2}^n G^{(k)}(x)<0
%\end{align} 

To show the expression for the second derivative observe that 
\begin{align}
i''(x; P_{X})&= G''(x) +\frac{1}{x}\\
&=\sum_{k=0}^\infty  P_{Y|X}(k|x)  \log  \left(\frac{1 }{ \bbE[X|Y=k+1]}  \right) -G'(x)+\frac{1}{x}\\
&=\sum_{k=0}^\infty  P_{Y|X}(k|x)  \log  \left(\frac{1 }{ \bbE[X|Y=k+1]}  \right) -i'(x; P_{X})+\log(x)+\frac{1}{x}.
\end{align}
\end{proof}

\section{Proof of Lemma~\ref{lem:bound_on_the_ratio}}
\label{app:lem:bound_on_the_ratio}

	We start from
	\begin{equation}
		\bbE \left[ \frac{ P_{Y|X}(\bar{Y}+1|\sfA)}{P_{Y^\star}(\bar{Y}+1)} \mathds{1}_{\bar{Y}< c}\right] = \sum_{k=0}^{c-1} P_{\bar{Y}}(k) \frac{\sfA^{k+1} \rme^{-\sfA}}{(k+1)! P_{Y^\star}(k+1)}. \label{eq:expbarY}
	\end{equation}
	We next find a lower bound on the $P_{Y^\star}(k+1)$ where we distinguish  between two regimes  $x_0 \le c $ and  $c>x_0$.
	
	First, for  $1\le k+1\le \min\{x_0,c\}$, we have the following lower bound on $P_{Y^\star}(k+1)$: 
	\begin{align}
		(k+1)! P_{Y^\star}(k+1) &= \sum_{0< x\le \sfA} P_{X^\star}(x) x^{k+1} \rme^{-x} \\
		&\ge \sum_{1\le x\le k+1} P_{X^\star}(x) x^{k+1} \rme^{-x} + \sum_{k+1<x\le x_0} P_{X^\star}(x) x^{k+1} \rme^{-x}  \\
		&\ge \sum_{1\le x\le k+1} P_{X^\star}(x)   \rme^{-1}+ \sum_{k+1<x\le x_0} P_{X^\star}(x) \sfA^{k+1} \rme^{-\sfA} \label{eq:Position_of_the_max} \\
		&\ge \sum_{1\le x\le k+1} P_{X^\star}(x)   \rme^{-1}+ \sum_{k+1<x\le x_0} P_{X^\star}(x) \sfA \rme^{-\sfA}  \label{eq:SomeProerties_That_come_from_monotonicity}\\
		&\ge \sfA \rme^{-\sfA} \sum_{1\le x\le x_0} P_{X^\star}(x) \label{eq:bound Aexp(-a)}\\
		&=\sfA \rme^{-\sfA} \left(1-\sum_{0\le x<1} P_{X^\star}(x)-P_{X^\star}(\sfA)\right) \\
		&\ge \sfA \rme^{-\sfA}\left(1-3\:\rme^{- \frac{C(\sfA)}{ 1-\rme^{-\sfA}}  }\right), \label{eq:lowPY}
	\end{align}
	where in  \eqref{eq:Position_of_the_max} we used $x^{k+1} \rme^{-x}\ge \rme^{-1}$, because $x \mapsto x^{k+1} \rme^{-x}$ is  increasing for $x<k+1$, and in the second summation we used that $x\mapsto x^{k+1} \rme^{-x}$ is  decreasing for $k+1<x\le x_0< \sfA$;  in \eqref{eq:SomeProerties_That_come_from_monotonicity}, we used $\sfA^{k+1}\ge \sfA$ for $\sfA\ge 1$ and $k\ge 0$; in \eqref{eq:bound Aexp(-a)}, we used $\sfA \rme^{-\sfA}\le \rme^{-1}$ for $\sfA\ge 1$; and in \eqref{eq:lowPY} we used that there is at most one support point in the open interval $(0,1)$, and the upper bound on the probability mass in \eqref{Eq:Bound_On_Probability_Values_poisson}. 
	
	If $c>x_0$, we can follow analogous bounding steps to write
	\begin{align}
	(k+1)! P_{Y^\star}(k+1) &\ge \sum_{1\le x\le x_0} P_{X^\star}(x)   \rme^{-1}\\
	&= \rme^{-1}\left(1-\sum_{0\le x<1} P_{X^\star}(x)-P_{X^\star}(\sfA)\right) \\
	&\ge \rme^{-1} \left(1-3\:\rme^{- \frac{C(\sfA)}{ 1-\rme^{-\sfA}}  }\right), \qquad x_0<k+1< c. \label{eq:lowPY_k_between_x0_A}
	\end{align}

	If $c\le x_0$, plugging~\eqref{eq:lowPY} into~\eqref{eq:expbarY}, we get
	\begin{align}
		\bbE \left[ \frac{ P_{Y|X}(\bar{Y}+1|\sfA)}{P_{Y^\star}(\bar{Y}+1)} \mathds{1}_{\bar{Y}< c}\right] &\le \frac{1}{\sfA \rme^{-\sfA}\left(1-3\:\rme^{- \frac{C(\sfA)}{ 1-\rme^{-\sfA}}  }\right)} \sum_{k=0}^{c-1} P_{\bar{Y}}(k) \sfA^{k+1}\rme^{-\sfA} \\
		&\le \frac{c  \sfA^{c-1}}{ 1-3\:\rme^{- \frac{C(\sfA)}{ 1-\rme^{-\sfA}}  }}  \\
		&\le \frac{x_0  \sfA^{x_0-1}}{ 1-3\:\rme^{- \frac{C(\sfA)}{ 1-\rme^{-\sfA}}  }}. \label{eq:Bound_on_ratio_c<x_0}
	\end{align}
	
	%	Let us now upper bound $P(\bar{Y}>c)$ with Markov's inequality:
	%	\begin{align}
	%	P(\bar{Y}>c) & \le \frac{\bbE [ \bar{Y}] }{c}\\
	%	&=   \frac{   x_0  }{c }.
	%	\end{align}
	
	Finally, if $c>x_0$, we have
	\begin{align}
	\bbE \left[ \frac{ P_{Y|X}(\bar{Y}+1|\sfA)}{P_{Y^\star}(\bar{Y}+1)} \mathds{1}_{\bar{Y}< c}\right] 
	&=\bbE \left[ \frac{ P_{Y|X}(\bar{Y}+1|\sfA)}{P_{Y^\star}(\bar{Y}+1)} (\mathds{1}_{\bar{Y}\le x_0}+\mathds{1}_{x_0<\bar{Y}< c})\right]\\
	&\le \frac{1}{1-3\:\rme^{- \frac{C(\sfA)}{ 1-\rme^{-\sfA}}  }}  \left[  x_0 \sfA^{x_0-1}+\sum_{k=\lceil x_0 \rceil}^{c-1} P_{\bar{Y}}(k) \sfA^{k+1}\rme^{-\sfA+1} \right] \label{eq:using_bound_x0<c}\\
	&\le \frac{1}{1-3\:\rme^{- \frac{C(\sfA)}{ 1-\rme^{-\sfA}}  }}  \left[  x_0 \sfA^{x_0-1}+\sum_{k=0}^{c-1} \frac{(x_0 \sfA)^k}{k!} \sfA \rme^{-\sfA+1-x_0} \right]\\
	&= \frac{1}{1-3\:\rme^{- \frac{C(\sfA)}{ 1-\rme^{-\sfA}}  }}  \left[  x_0 \sfA^{x_0-1}+\sfA \rme^{-\sfA+1-x_0+x_0\sfA}\frac{\Gamma(c,x_0\sfA)}{c!}  \right]\label{eq:incompleteGamma}\\
	&\le \frac{1}{1-3\:\rme^{- \frac{C(\sfA)}{ 1-\rme^{-\sfA}}  }}  \left[  x_0 \sfA^{x_0-1}+\sfA \rme^{-\sfA+1-x_0}\frac{(x_0 \sfA)^{c} }{c!(x_0 \sfA-c+1)}  \right] \label{eq:incompleteGamma_bound} \\
	&\le \frac{1}{1-3\:\rme^{- \frac{C(\sfA)}{ 1-\rme^{-\sfA}}  }}  \left[  x_0 \sfA^{x_0-1}+\sfA \rme^{-\sfA+1-x_0+c}\frac{1 }{\sqrt{2\pi c}(x_0 \sfA-c+1)} \left(\frac{x_0 \sfA}{c}\right)^c \right], \label{eq:stirlingbound}
	\end{align}
	where \eqref{eq:using_bound_x0<c} follows from the bound in \eqref{eq:Bound_on_ratio_c<x_0};
	\eqref{eq:incompleteGamma} and \eqref{eq:incompleteGamma_bound} follow by introducing the incomplete Gamma function $\Gamma(c,x_0\sfA)$, which is then upper-bounded as~\cite[Thm.~2.1]{borwein2009uniform}
	\begin{equation}
	\Gamma(c,x_0\sfA) \le \frac{(x_0 \sfA)^{c} \rme^{-x_0 \sfA}}{x_0 \sfA-c+1}, \qquad  0<c<x_0\sfA+1;
	\end{equation}
	and \eqref{eq:stirlingbound} follows by using Stirling's lower bound $c!\ge \sqrt{2\pi c}\: c^{c}\:\rme^{-c}$.

	\section{Proof of Lemma~\ref{lem:bound_poisson_tail}}
	\label{app:lem:bound_poisson_tail}
 By using the bound on the tail of the  Poisson random variable in \cite[Thm.~5.4]{mitzenmacher2017probability}, we have that
\begin{equation}
			\bbP \left[ \overline{Y} \ge c \right] \le \frac{(\rme x_0)^c \rme^{-x_0}}{ c^c}, \quad c>x_0.  \label{eq:consequence_chernoff}
			\end{equation} 
By choosing $c=\rme x_0 \log(\sfA)>x_0$ for $\sfA\ge \rme$, the bound in \eqref{eq:consequence_chernoff} implies that 
\begin{align}
			\bbP \left[ \overline{Y} \ge \rme x_0 \log(\sfA) \right]&\le \left( \frac{1 }{ \log(\sfA)} \right)^{ \rme x_0 \log(\sfA)} \rme^{-x_0} \\
			&\le \left( \frac{1 }{ \log(\sfA)} \right)^{ \rme  \log(\sfA)} \rme^{-1},
			\end{align} 
			where the last inequality follows from the assumption that $\sfA \ge \rme$ and $x_0 \ge 1$. 

The proof is concluded by observing that $\sfA \ge \rme$ implies
\begin{equation}
 \left( \frac{1 }{ \log(\sfA)} \right)^{ \rme  \log(\sfA)}  \rme^{-1} \le \frac{1}{\sfA}.
\end{equation} 

\end{appendices}

\section*{Acknowledgement}
The authors would like to thank Semih Yagli for his valuable comments.

\bibliographystyle{IEEEtran}
\bibliography{biblio}

\end{document}